\documentclass[twoside,11pt]{article}
\usepackage{jair}
\usepackage{theapa}
\usepackage{epsfig,subfigure}
\usepackage{amssymb,amsmath}
\usepackage[usenames]{color}
\usepackage{pst-tree, pst-node, subfigure}
\usepackage{algorithmicx}
\usepackage{algpseudocode}
\usepackage[normalem]{ulem}

\newtheorem{theorem}{Theorem}
\newtheorem{lemma}{Lemma}

\newtheorem{definition}{Definition}

\newenvironment{proof}[1][]
  {\noindent{\bf Proof#1:}}{\hfill $\Box$\smallskip}

\catcode`@=11

% ...

\def\imparray{\stepcounter{equation}\let\@currentlabel=\theequation
\global\@eqnswtrue
\global\@eqcnt\z@\tabskip\@centering\let\\=\@eqncr
$$\halign to \displaywidth\bgroup\llap{${##}$\hskip 4\arraycolsep}\tabskip\z@&
  \@eqnsel\hskip\@centering
  $\displaystyle\tabskip\z@{##}$&\global\@eqcnt\@ne 
  \hskip 2\arraycolsep \hfil${##}$\hfil
  &\global\@eqcnt\tw@ \hskip 2\arraycolsep $\displaystyle\tabskip\z@{##}$\hfil 
   \tabskip\@centering&\llap{##}\tabskip\z@\cr}

\def\endimparray{\@@eqncr\egroup
      \global\advance\c@equation\m@ne$$\global\@ignoretrue}

\@namedef{imparray*}{\def\@eqncr{\nonumber\@seqncr}\imparray}
\@namedef{endimparray*}{\nonumber\endimparray}

\def\ps@headerfooter{\let\@mkboth\@gobbletwo
     \def\@oddhead{\reset@font\@myhead}
     \def\@oddfoot{\reset@font\@myfoot}
     \let\@evenhead\@oddhead
     \let\@evenfoot\@oddfoot}

\def\header#1{\def\@myhead{#1}}
\def\footer#1{\def\@myfoot{#1}}

\catcode`@=12

\newcommand{\bcm}{\begin{center}\begin{math}}
\newcommand{\ecm}{\end{math}\end{center}}

\newcommand{\bi}{\begin{itemize}}
\newcommand{\ei}{\end{itemize}}

\newcommand{\be}{\begin{enumerate}}
\newcommand{\ee}{\end{enumerate}}

\newcommand{\tail}[2]{\mathchoice
{\left(\!\!\left({{#1}\atop{#2}}\right)\!\!\right)}
{\mathopen {\hbox{$\delimitershortfall=0pt
\left(\vcenter to \the\fontdimen21\textfont2{}\right.
\nulldelimiterspace=0pt \mathsurround=0pt$}}\mskip-4mu
\mathopen  {\hbox{$\delimitershortfall=0pt
\left(\vcenter to \the\fontdimen21\textfont2{}\right.
\nulldelimiterspace=0pt \mathsurround=0pt$}}
\mskip-1.9mu {{#1}\atop{#2}} \mskip-1.9mu
\mathclose {\hbox{$\delimitershortfall=0pt
\left)\vcenter to \the\fontdimen21\textfont2{}\right.
\nulldelimiterspace=0pt \mathsurround=0pt$}} \mskip-4mu
\mathclose {\hbox{$\delimitershortfall=0pt
\left)\vcenter to \the\fontdimen21\textfont2{}\right.
\nulldelimiterspace=0pt \mathsurround=0pt$}}}
{\mathopen {\hbox{$\delimitershortfall=0pt
\left(\vcenter to \the\fontdimen21\scriptfont2{}\right.
\nulldelimiterspace=0pt \mathsurround=0pt$}}\mskip-4mu
\mathopen  {\hbox{$\delimitershortfall=0pt
\left(\vcenter to \the\fontdimen21\scriptfont2{}\right.
\nulldelimiterspace=0pt \mathsurround=0pt$}}
\mskip-2.75mu {{#1}\atop{#2}} \mskip-2.75mu
\mathclose {\hbox{$\delimitershortfall=0pt
\left)\vcenter to \the\fontdimen21\scriptfont2{}\right.
\nulldelimiterspace=0pt \mathsurround=0pt$}} \mskip-4mu
\mathclose {\hbox{$\delimitershortfall=0pt
\left)\vcenter to \the\fontdimen21\scriptfont2{}\right.
\nulldelimiterspace=0pt \mathsurround=0pt$}}}
{\mathopen {\hbox{$\delimitershortfall=0pt
\left(\vcenter to \the\fontdimen21\scriptscriptfont2{}\right.
\nulldelimiterspace=0pt \mathsurround=0pt$}}\mskip-4mu
\mathopen  {\hbox{$\delimitershortfall=0pt
\left(\vcenter to \the\fontdimen21\scriptscriptfont2{}\right.
\nulldelimiterspace=0pt \mathsurround=0pt$}}
\mskip-2.75mu {{#1}\atop{#2}} \mskip-2.75mu
\mathclose {\hbox{$\delimitershortfall=0pt
\left)\vcenter to \the\fontdimen21\scriptscriptfont2{}\right.
\nulldelimiterspace=0pt \mathsurround=0pt$}} \mskip-4mu
\mathclose {\hbox{$\delimitershortfall=0pt
\left)\vcenter to \the\fontdimen21\scriptscriptfont2{}\right.
\nulldelimiterspace=0pt \mathsurround=0pt$}}}
}

\newcommand{\comment}[1]{}

\def\sfrac#1#2{\leavevmode\kern.1em
\raise.5ex\hbox{\the\scriptfont0 #1}\kern-.1em
/\kern-.15em\lower.25ex\hbox{\the\scriptfont0 #2}}

\newtheorem{example}{Example}
\newcommand{\etal}{{et al.}} 
\newcommand{\chain}{{\sc Chain}} 

{\makeatletter
 \gdef\xxxmark{%
   \expandafter\ifx\csname @mpargs\endcsname\relax 
     \expandafter\ifx\csname @captype\endcsname\relax 
       \marginpar{xxx}
     \else
       xxx 
     \fi
   \else
     xxx 
   \fi}
 \gdef\xxx{\@ifnextchar[\xxx@lab\xxx@nolab}
 \long\gdef\xxx@lab[#1]#2{{\bf [\xxxmark #2 ---{\sc #1}]}}
 \long\gdef\xxx@nolab#1{{\bf [\xxxmark #1]}}
}

\long\def\symbolfootnote[#1]#2{\begingroup%
\def\thefootnote{\fnsymbol{footnote}}\footnote[#1]{#2}\endgroup}

\jairheading{30}{2007}{133--179}{3/07}{9/07}
\firstpageno{133}

\ShortHeadings{Chain: An Online Double Auction}
{Bredin, Parkes and Duong}

\begin{document}
\title{Chain: A Dynamic Double Auction Framework for Matching Patient Agents}
\author{\name Jonathan Bredin
\email bredin@acm.org\\
\addr Dept. of Mathematics and Computer Science, Colorado College\\
Colorado Springs, CO 80903, USA\\
\name David C. Parkes 
\email parkes@eecs.harvard.edu\\
\name Quang Duong
\email qduong@fas.harvard.edu\\
\addr School of Engineering and Applied Sciences, Harvard University\\
Cambridge, MA 02138, USA}
\maketitle

\begin{abstract}
In this paper we present and evaluate a general framework for the design of
truthful auctions for matching agents in a dynamic, two-sided market.
A single commodity, such as a resource or a task, is bought and sold
by multiple buyers and sellers that arrive and depart over time. 
Our algorithm, \chain, provides the first framework that allows
a truthful dynamic double auction (DA) to be constructed from a 
truthful, single-period 
(i.e. static) double-auction rule. The pricing and matching
method of the \chain\ construction is unique amongst dynamic-auction 
rules that adopt the same building block.
We examine experimentally the allocative
efficiency of \chain\ when instantiated on various single-period rules,
including the canonical McAfee double-auction rule.
For a baseline we also consider non-truthful double auctions populated
with ``zero-intelligence plus"-style learning
agents. \chain-based auctions perform well
in comparison with other schemes, especially  
as arrival intensity falls and agent valuations become more volatile. 
\end{abstract}



\section{Introduction}

Electronic markets are increasingly popular as a method to facilitate
increased efficiency in the supply chain, with firms using markets to
procure goods and services. Two-sided markets facilitate
trade between many buyers and many sellers and find application
to trading diverse resources, including bandwidth, securities and pollution rights. 
Recent years have also brought
increased attention to resource allocation in the context of on-demand
computing  and 
grid computing.
Even within settings of
cooperative coordination, such as those of multiple robots,
researchers have turned to auctions as methods for task allocation and
joint exploration~\shortcite{gerkey:sold,lagoudakis:robotroute,lin:combtask}. 

In this paper we consider a dynamic two-sided market for 
a single commodity, for instance a unit of a resource
(e.g. time on a computer, some quantity of memory chips) or a task to
perform (e.g. a standard database query to execute, a location
to visit). Each agent, whether buyer or seller, arrives 
dynamically and needs to be 
matched within a time interval. Cast as a task-allocation problem,  
a seller can perform the task when allocated within some time 
interval and incurs a cost when assigned. A buyer has positive
value for the task being assigned (to any seller) within some time
interval. The arrival time, acceptable time interval, and value
(negative for a seller) for a trade are all private information to
an agent. Agents are self-interested and can choose to misrepresent
all and any of this information to the market in order to obtain
a more desirable price.

The matching problem combines elements of online algorithms and sequential
decision making with considerations from mechanism
design. Unlike traditional sequential decision making, a protocol for this
problem must provide incentives for agents to report truthful
information to a match-maker. Unlike traditional mechanism design,
this is a dynamic problem with agents that arrive and leave over time.
We model this problem as a dynamic double auction (DA) for identical items.
The match-maker becomes the auctioneer. Each seller brings a task
to be performed during a time window and each buyer brings the
capability to perform a single task. The double-auction setting also is
of interest in its own right as a protocol for matching in a
dynamic business-to-business exchange. 

Uncertainty about the future coupled with the two-sided nature of the market
leads to an interesting  mechanism design problem.
For example, consider the scenario where the auctioneer must decide how
(and whether) to match a seller with reported cost of \$6
at the end of its time interval
with a present and unmatched
buyer, one of which has a reported value of \$8 and one a 
reported value of \$9.
Should the auctioneer pair the higher bidder with the seller?
What happens if a seller, willing to sell for \$4, arrives after the 
auctioneer acts upon the matching decision?
How should the matching algorithm be designed so that no agent can benefit
from misstating its earliest arrival, latest departure, or value for a trade?

\chain\ provides a general framework that allows a truthful
dynamic double auction to be constructed from a truthful,
single-period (i.e. static) double-auction rule. 
The auctions constructed by \chain\ are truthful, in the sense that 
the dominant strategy for an agent, whatever the future
auction dynamics and bids from other agents, is to report
its true value for a trade (negative if selling) 
and true patience (maximal tolerance for trade delay)
immediately upon arrival into the market.
We also allow for
randomized mechanisms and, in this case, require {\em strong}
truthfulness: the DA should be truthful for all possible random coin
flips of the mechanism.
One of the DAs in the class of auctions implied by \chain\ is a dynamic
generalization of McAfee's~\citeyear{mcafee:double92} canonical
truthful, no-deficit auction for a single period. Thus, we provide the
first examples of truthful, dynamic DAs that allow for dynamic price competition between buyers and
sellers.\footnote{The closest work in the literature is due to 
Blum \etal~\citeyear{blum06}, who present a
truthful, dynamic DA for our model that matches bids and asks
based on a price sampled from some bid-independent distribution.
We compare the performance of our schemes with this scheme in
Section~\ref{sec:empirical}.}

The main technical challenge presented by dynamic DAs
is to provide truthfulness without incurring a budget deficit, 
while handling uncertainty about future trade opportunities. 
Of particular concern is to ensure that
an agent does not indirectly affect its price through
the effect of its bid on the prices faced by other agents
and thus other supply and demand in the market. We need to preclude this
 because the {\em availability} of
trades depends on the price faced by other agents. For example, a
buyer that is required to pay \$4 in the DA to trade might like to
decrease the price that a potentially matching seller will receive
from \$6 to \$3 to allow for trade.

\chain\ is a modular approach to auction design, 
which takes as a building block a single-period
matching rule and provides a method to invoke the rule in each
of multiple periods while also providing for truthfulness. 
We characterize properties that a {\em well-defined} single-period matching rule must
satisfy in order for \chain\ to be truthful. 
We further identify the technical
property of {\em strong no-trade}, with which we can
isolate agents that fail to 
trade in the current period but can nevertheless survive and be 
eligible to trade in a future period. An auction designer defines
the strong no-trade predicate, in addition to providing a well-defined
single-period matching rule.
Instances within this class include those constructed in terms of
both ``price-based" matching rules and ``competition-based" matching rules.
Both can depend on history and be adaptive, but only
the competition-based rules use the active bids and asks to 
determine the prices in the current
period, facilitating a more direct competitive processes. 

In proving that \chain, when combined with a well-defined
matching rule and a valid strong no-trade predicate, is truthful
we leverage a recent price-based characterization for truthful
online mechanisms~\shortcite{hajiaghayi05}. We also show that
the pricing and matching rules defined by 
\chain\ are unique amongst the family of mechanisms 
that are constructed with a single-period matching rule
as a building block. 
Throughout our work we assume that a constant limits every
buyer and seller's patience. To motivate
this assumption we provide a simple
environment in which {\em no} truthful, no-deficit DA can implement
some constant fraction of the number of the efficient trades,
for any constant.

We adopt {\em allocative
efficiency} as our design objective, 
which is to say auction protocols that maximize the
expected total value from the sequence of trades. We also consider
{\em net efficiency}, wherein any net outflow of payments to the
marketmaker is also accounted for in considering the quality of a 
design.
Experimental results explore the allocative efficiency of 
\chain\ when instantiated to various single-period matching
rules and for a range of different
  assumptions about market volatility and maximal patience. 
For a baseline we consider the efficiency of a  standard
  (non-truthful) open outcry DA populated with simple adaptive trading
  agents modeled after ``zero-intelligence plus" (ZIP)
agents~\shortcite{cliff:simple,preist:zip}. We also compare
the efficiency of \chain\ with that of a truthful online DA due to 
Blum \etal~\citeyear{blum06}, which selects a fixed trading price
to guarantee competitiveness in an adversarial model.

From within the truthful mechanisms we find that adaptive, price-based
instantiations of \chain\ are the most effective for high arrival
intensity and low volatility. Even defining a single, well-chosen
price that is optimized for the market conditions can be reasonably
effective in promoting efficient trades in low volatility environments.
On the other hand, for medium to low arrival intensity and medium to
high volatility we find that the \chain-based DAs that allow for
dynamic price competition, such as the McAfee-based rule, are 
most efficient. The same qualitative observations hold whether one is
interested in allocative efficiency or net efficiency, although
the adaptive, price-based methods have better performance in terms of
net efficiency.
The Blum \etal~\citeyear{blum06} rule fairs
poorly in our tests, which is perhaps unsurprising given that it
is optimized for worst-case performance in an adversarial setting.
When populated with ZIP agents, we find that non-truthful DAs can 
provide very good efficiency in low volatility environments but
poor performance in high volatility environments.
The good performance of the ZIP-based market occurs when agents learn
to bid approximately truthfully; i.e., when the market operates
as if truthful, but without incurring the stringent cost (e.g., through trading constraints) of imposing truthfulness explicitly. An equilibrium analysis is 
available only for the truthful DAs; we
have no way of knowing how close the ZIP agents are 
to playing an equilibrium, and note that the ZIP agents do not even {\em consider} 
performing time-based manipulations.

\vspace{-0.1cm}
\subsection{Outline}

Section~\ref{sec:dynDA} introduces the
dynamic DA model, including our assumptions, and presents desiderata for online
DAs and a price-based characterization for the design of truthful
dynamic auctions. 
Section~\ref{sec:framework} defines the \chain\ algorithm together
with the building block of a well-defined, single-period matching rule
and the strong no-trade predicate. Section~\ref{sec:practical} 
gives a number of instantiations to both price-based and competition-based
matching rules, including a general method to define the 
strong no-trade predicate given a price-based
instantiation. Section~\ref{sec:theory} proves truthfulness,
no-deficit and feasibility of the \chain\ auctions and also
establishes their uniqueness amongst auctions constructed from 
the same single-period matching-rule building block. The importance
of the assumption about maximal agent patience is established. 
Section~\ref{sec:empirical} presents
our empirical analysis, including a description of the simple adaptive
agents that we use to populate a non-truthful open-outcry DA and
provide a benchmark. Section~\ref{sec:related} gives related work.
In Section~\ref{sec:conclude} we conclude with a discussion about the merits of truthfulness in markets 
and present possible extensions.

\vspace{-0.2cm}
\section{Preliminaries: Basic Definitions}
\label{sec:dynDA}

Consider a dynamic auction model with discrete, possibly infinite, time periods
$T=\{1,2,\ldots\}$, indexed by $t$.
The double auction (DA) provides a market for a single
commodity. Agents are either buyers or sellers interested
in trading a single unit of the commodity. An agent's type,
$\theta_i=(a_i,d_i,w_i)\in \Theta_i$, where $\Theta_i$ is the set of
possible types for agent $i$, defines an arrival $a_i$, departure
$d_i$, and value $w_i\in \mathbb{R}$ for trade. If the agent is a
buyer, then $w_i>0$. If the agent is a seller, then $w_i\leq 0$. 
We assume a maximal patience $K$, so that $d_i\leq a_i+K$
for all agents.

The arrival time models the first time at which an agent learns about
the market or learns about its value for a trade. Thus, information 
about its type is not available before period $a_i$ (not even to agent $i$) and
the agent cannot engage in trade before period $a_i$. The departure
time, $d_i$, 
models the final period in which a buyer has positive value for a
trade, or the final period in which a seller is willing to 
engage in trade. We model risk-neutral agents with quasi-linear utility,
$w_i-p$ when a trade occurs in $t\in[a_i,d_i]$ 
and payment $p$ is collected (with $p<0$ if  the agent is a seller).
Agents are rational and self-interested, and act to maximize expected
utility. By assumption, sellers have no utility for payments
received after their true departure period. 

Throughout this paper we adopt {\em bid} to refer, generically, to a
claim that an agent -- either a buyer or a seller -- makes to a DA about
its type. In addition, when we need to be specific about the
distinction between claims made by buyers and claims made by sellers 
we refer to the {\em bid} from a buyer and the {\em ask} from a
seller. 

\vspace{-0.1cm}
\subsection{Example}

Consider the following naive generalization of the (static) trade-reduction
DA~\shortcite{lavi05,mcafee:double92} to this dynamic environment. 
A bid 
from an agent is a claim about its type
$\hat{\theta}_i=(\hat{a}_i,\hat{d}_i,\hat{w}_i)$, necessarily made in
period $t=\hat{a}_i$. Bids are {\em active} while
$t\in[\hat{a}_i,\hat{d}_i]$ and no trade has occurred. 

Then in each period $t$, use the trade-reduction DA to determine
which (if any) of the active bids trade and at what price. These
trades occur immediately. 
The trade-reduction DA (tr-DA) works as follows:
Let $B$ denote the set of bids and $S$ denote the set of asks.
Insert a dummy bid with value $+\infty$ into $B$ and a dummy ask
with value $0$ into $S$. When $|B|\geq 2$ and $|S|\geq 2$ then 
sort $B$ and $S$ in order of decreasing
value. Let $\hat w_{b_0}\geq \hat w_{b_1}\geq \ldots$ and $\hat w_{s_0}\geq \hat w_{s_1}\geq \ldots$ 
denote the bid and ask values with $(b_0,s_0)$ denoting the dummy
bid-ask pair. Let $m\geq 0$ index the last pair of bids and asks to 
clear in the efficient trade, such that $\hat w_{b_m}+\hat w_{s_m}\geq 0$ and
$\hat w_{b_{m+1}}+\hat w_{s_{m+1}}<0$. When $m\geq 2$ then bids $\{b_1,\ldots,b_{m-1}\}$ and
asks $\{s_1,\ldots,s_{m-1}\}$ trade and payment $\hat w_{b_m}$ is collected 
from each winning buyer and payment $-\hat w_{s_m}$ is made to each 
winning seller.

First consider a static tr-DA with the following bids and asks:

\vspace{1mm}
\centerline{\begin{tabular}{cc|cc}
\multicolumn{2}{c|}{\bf B} & \multicolumn{2}{c}{\bf S} \\
$i$ & $\hat w_i$ & $i$ & $\hat w_i$ \\\hline
$b_1^\ast$ & 15 & $s_1^\ast$ & -1\\
$b_2^\ast$ & 10 & $s_2^\ast$ & -1\\
$b_3^\ast$ & 4 & $s_3^\ast$ & -2\\
$b_4$ & 3 & $s_4$ & -2 \\
\hline
& & $s_5$ & -5
\end{tabular}}
\vspace{1mm}

The line indicates that bids (1--4) and asks (1--4) could be matched for
efficient trade. By the rules of the tr-DA, bids (1--3) and asks
(1--3) trade, with payments \$3 collected from winning buyers and
payment \$2 made to winning sellers. The auctioneer earns a profit
of \$3. The asterisk notation indicates the bids and asks that trade.
The tr-DA is {\em truthful}, in the sense that it is a
dominant-strategy for every agent to report its true value whatever
the reports of other agents. For intuition, consider the buy-side. The
payment made by winners is independent of their bid price while the
losing bidder could only win by bidding more than \$4, at which point
his payment would be \$4 and more than his true value.

Now consider a dynamic variation with buyer types
$\{(1,2,15),(1,2,10),(1,2,4),(2,2,3)\}$ and seller types
$\{(1,2,-1),(2,2,-1),(1,1,-2),(2,2,-2),(1,2,-5)\}$. When agents are
truthful, the dynamic tr-DA plays out as follows:

\vspace{1mm}
\centerline{\begin{tabular}{cc|cc c c cc|cc}
\multicolumn{4}{c}{period 1} & & & \multicolumn{4}{c}{period 2}
\\
\multicolumn{2}{c|}{\bf B} & \multicolumn{2}{c}{\bf S} & & & \multicolumn{2}{c|}{\bf B} & \multicolumn{2}{c}{\bf S}
\\
$i$ & $\hat w_i$ & $i$ & $\hat w_i$ & & & $i$ & $\hat w_i$ & $i$ & $\hat w_i$
\\ \cline{1-4} \cline{7-10} 
$b_1^\ast$ & 15 & $s_1^\ast$ & -1 & & & $b_2^\ast$ & 10 & $s_2^\ast$ & -1
\\
$b_2$ & 10 & $s_3$ & -2 & & & $b_3$ & 4 & $s_4$ & -2 
\\ \cline{1-4} \cline{7-10} 
$b_3$ & 4 & $s_5$ & -5 & & & $b_4$ & 3 & $s_5$ & -5 
\end{tabular}}
\vspace{1mm}

In period 1, buyer 1 and seller 1 trade at payments of \$10 and \$2 
respectively. In period 2, buyer 2 and seller 2 trade at payments of
\$4 and \$2 respectively. 
But now we can construct two kinds of
manipulation to show that this dynamic DA is not truthful. 
First, buyer 1 can do better by delaying his reported arrival until 
period 2:

\vspace{1mm}
\centerline{\begin{tabular}{cc|cc c c cc|cc}
\multicolumn{4}{c}{period 1} & & & \multicolumn{4}{c}{period 2}
\\
\multicolumn{2}{c|}{\bf B} & \multicolumn{2}{c}{\bf S} & & & \multicolumn{2}{c|}{\bf B} & \multicolumn{2}{c}{\bf S}
\\
$i$ & $\hat w_i$ & $i$ & $\hat w_i$ & & & $i$ & $\hat w_i$ & $i$ & $\hat w_i$
\\ \cline{1-4} \cline{7-10} 
$b_2^\ast$ & 10 & $s_1^\ast$ & -1 & & & $b_1^\ast$ & 15 & $s_2^\ast$ & -1
\\
$b_3$ & 4 & $s_3$ & -2 & & & $b_3$ & 4 & $s_4$ & -2 
\\ \cline{1-4} \cline{7-10} 
& & $s_5$& -5 & & & $b_4$ & 3 & $s_5$ & -5 
\end{tabular}}
\vspace{1mm}

\noindent Now, buyer 2 trades in period 1 and does not set the price to buyer 1
in period 2. Instead, buyer 1 now trades in period 2 and makes payment
\$4. 

Second, buyer 3 can do better by increasing his reported value:

\vspace{1mm}
\centerline{\begin{tabular}{cc|cc c c cc|cc}
\multicolumn{4}{c}{period 1} & & & \multicolumn{4}{c}{period 2}
\\
\multicolumn{2}{c|}{\bf B} & \multicolumn{2}{c}{\bf S} & & & \multicolumn{2}{c|}{\bf B} & \multicolumn{2}{c}{\bf S}
\\
$i$ & $\hat w_i$ & $i$ & $\hat w_i$ & & & $i$ & $\hat w_i$ & $i$ & $\hat w_i$
\\ \cline{1-4} \cline{7-10} 
$b_1^\ast$ & 15 & $s_1^\ast$ & -1 & & & $b_3^\ast$ & 6 & $s_2^\ast$ & -1
\\
$b_2^\ast$ & 10 & $s_3^\ast$ & -2 & & & $b_4$ & 3 & $s_4$ & -2 
\\ \cline{7-10} 
$b_3$ & 6 & $s_5$ & -5 & & & & & $s_5$ & -5 \\ \cline{1-4} 
\end{tabular}}
\vspace{1mm}

Now, buyers 1 and 2 both trade in period 1 and this allows buyer 3 to win
(at a price below his true value) in period 2. This is a particularly
interesting manipulation because the agent's manipulation is by {\em
increasing} its bid above its true value. By doing so, it allows more
trades to occur and makes the auction less competitive in the next
period. 

\vspace{-0.1cm}
\subsection{Dynamic Double Auctions: Desiderata}

We consider only direct-revelation, dynamic DAs that restrict
the message that an agent can send to the auctioneer to a single, direct
claim about its type. We also consider ``closed" auctions so that an
agent receives no feedback before reporting its type and cannot
condition its strategy on the report of another agent.\footnote{The 
restriction to direct-revelation, online mechanisms is without loss of
generality when combined with a simple heart-beat message from an
agent to indicate its presence in any period $t$ during its reported
arrival-departure interval. See the work of Pai and Vohra~\citeyear{pai06} and
Parkes~\citeyear{parkes07}.} 

Given this, let $\theta^t$ denote the set of agent 
types reported in period $t$, $\theta=(\theta^1,\theta^2,\ldots, 
\theta^t,\ldots, )$ denote a complete type profile
(perhaps unbounded), and $\theta^{\leq t}$ denote the type profile
restricted to agents with (reported) arrival no later than period
$t$. A report $\hat{\theta}_i=(\hat{a}_i,\hat{d}_i,\hat{w}_i)$
represents a {\em commitment} to buy (sell) one unit of the commodity
in any period $t\in[\hat{a}_i,\hat{d}_i]$ for a payment of at most
$\hat{w}_i$. Thus, if a seller reports a departure time
$\hat{d}_i>d_i,$ it must commit to complete a trade that
occurs after her true departure and even though a seller is modeled
as having no utility for payments received after her true
departure.

A dynamic DA, $M=(\pi,x)$, defines an allocation policy
$\pi=\{\pi^t\}^{t\in T}$ and payment policy $x=\{x^t\}^{t\in T}$,
where $\pi^t_i(\theta^{\leq t})\in \{0,1\}$ indicates whether or 
not agent $i$ trades in period $t$ given reports $\theta^{\leq t},$ 
and $x^t_i(\theta^{\leq t})\in\mathbb{R}$ indicates a payment made by
agent $i$, negative if this is a payment received by the agent.
The auction rules can also be {\em stochastic}, so that $\pi_i^t(\theta^{\leq
t})$ and $x_i^t(\theta^{\leq t})$ are random variables. 
For a dynamic DA to be well defined, it must hold that 
$\pi^t_i(\theta^{\leq t})=1$ in at most
one period $t\in[a_i,d_i]$ and zero otherwise, and the payment
collected from agent $i$ is zero except in periods $t\in[a_i,d_i]$.

In formalizing the desiderata for dynamic DAs, 
it will be convenient to adopt $(\pi(\theta),x(\theta))$ to denote the complete
sequence of allocation decisions given reports $\theta$, with
shorthand $\pi_i(\theta)\in\{0,1\}$ and $x_i(\theta)\in\mathbb{R}$ to
indicate whether agent $i$ trades during its reported
arrival-departure interval, and the total payment made by agent $i$,
respectively. 
By a slight abuse of notation, we write $i\in\theta^{\leq t}$ to
denote that agent $i$ reported a type no later than period $t$. Let
$B$ denote the set of buyers and $S$ denote the set of sellers. 

We shall require that the
dynamic DA satisfies {\em no-deficit}, {\em feasibility}, 
{\em individual-rationality} and {\em truthfulness}. No-deficit
ensures that the auctioneer has a cash surplus in every period:
\begin{definition}[no-deficit]
A dynamic DA, $M=(\pi,x)$ is no-deficit if:
\begin{align}
\sum_{i\in \theta^{\leq t}}
\sum_{t'\in[a_i,\min(t,d_i)]}\!\!\!x_i^{t'}(\theta^{\leq t'})&\geq 0
, \ \ \forall t, \forall
\theta
\end{align}
\end{definition}

Feasibility ensures that the auctioneer does not need to
take a short position in the commodity traded in the market in any
period:
\begin{definition}[feasible trade]
A dynamic DA, $M=(\pi,x)$ is feasible if:
\begin{align}
\sum_{i\in \theta^{\leq t}, i\in S}
\sum_{t'\in[a_i,\min(t,d_i)]}\!\!\!\pi_i^{t'}(\theta^{\leq t'})
-\sum_{i\in \theta^{\leq t}, i\in B}
\sum_{t'\in[a_i,\min(t,d_i)]}\!\!\!\pi_i^{t'}(\theta^{\leq t'})
& \geq 0, \ \ \forall t, \forall
\theta
\end{align}
\end{definition}

This definition of feasible trade assumes 
that the auctioneer can ``hold" an item that 
is matched between a seller-buyer pair, for instance only releasing it
to the buyer upon his reported departure.
See the remark concluding this section for a discussion of this assumption.

Let $v_i(\theta_i,\pi(\theta'_i,\theta_{-i}))\in\mathbb{R}$ denote the 
value of an agent with type $\theta_i$ for the allocation decision
made by policy $\pi$ given report $(\theta'_i,\theta_{-i})$,
i.e. $v_i(\theta_i,\pi(\theta'_i,\theta_{-i}))=w_i$ if the agent
trades in period $t\in[a_i,d_i]$ and $0$ if it trades outside of this
interval and is a buyer, or $-\infty$ if it trades outside of this
interval and is a seller. Individual-rationality requires that agent $i$'s utility is
non-negative when it reports its true type, whatever the reports of
other agents: 
\begin{definition}[individual-rational]
A dynamic DA, $M=(\pi,x)$ is individual-rational (IR) if
$v_i(\theta_i,\pi(\theta))-x_i(\theta)\geq 0$ for all $i$, all 
$\theta$.
\end{definition}

In order to define truthfulness, we introduce notation
$C(\theta_i)\subseteq \Theta_i$ for $\theta_i\in \Theta_i$ to denote
the set of {\em available misreports} to an agent with true type
$\theta_i$. In the standard model adopted in offline mechanism design, it is
typical to assume $C(\theta_i)=\Theta_i$ with all misreports
available. Here, we shall assume {\em
no early-arrival} misreports, with
$C(\theta_i)=\{\hat{\theta}_i=(\hat{a}_i,\hat{d}_i,\hat{w}_i) :
a_i\leq \hat{a}_i\leq \hat{d}_i\}$. This assumption of limited
misreports is adopted in earlier work on online mechanism
design~\shortcite{hajiaghayi:online04}, and is 
well-motivated when the arrival time is the first period in which
a buyer first decides to acquire an item or the period in which a
seller first decides to sell an item.

\begin{definition}[truthful]
Dynamic DA, $M=(\pi,x)$, is dominant-strategy incentive-compatible, or
truthful, given limited misreports $C$ if:\sloppy
\begin{align*}
v_i(\theta_i,\pi(\theta_i,\theta'_{-i}))-x_i(\theta_i,\theta'_{-i})&\geq
v_i(\theta_i,\pi(\hat{\theta}_i,\theta'_{-i}))-x_i(\hat{\theta}_i,\theta'_{-i}).
\end{align*}
for all $\hat{\theta}_i\in C(\theta_i)$, all $\theta_i$, all
$\theta'_{-i}\in C(\theta_{-i})$, all $\theta_{-i}\in\Theta_{-i}$.
\end{definition}

This is a robust equilibrium concept: an agent maximizes its utility
by reporting its true type whatever the reports of other agents. 
Truthfulness is useful because it simplifies the decision problem facing 
bidders: an agent can determine its optimal bidding strategy
without a model of either the auction dynamics or the other agents.
In the case that the allocation and payment policy is {\em stochastic}, then
we adopt the requirement of {\bf strong truthfulness} so that an agent
maximizes its utility whatever the random sequence of coin flips
within the auction.

\paragraph{Remark.}
\label{sec:remark}
The flexible definition of feasibility, in which the auctioneer is
able to take a long position in the commodity, allows the auctioneer
to time trades by receiving the unit sold by a seller in one period but only releasing
it to a buyer in a later period. This allows for
truthfulness in environments in which bidders can overstate their
departure period. In some settings this is an unreasonable requirement, however, for
instance when the commodity 
represents a task that is performed, or because a physical good is being traded
in an electronic market.\footnote{Note that if the
task is a computational task, then tasks can be handled within this
model by requiring that the seller performs the task when it is
matched but with a commitment to hold onto the result 
until the matched buyer is ready to
depart.}  In these cases, the definition of feasibility
strengthened to require exact trade-balance in every
period. The tradeoff is that available misreports must be further
restricted, with agents limited to reporting {\em no late-departures}
in addition to {\em no early-arrivals}~\shortcite{lavi05,hajiaghayi05}.  
For the rest of the paper
we work in the ``relaxed feasibility, no early-arrival" model. The
\chain\ framework can be immediately extended to the ``strong-feasibility, no
early-arrival and no late-departure" model by executing trades
immediately rather than delaying the trade until a buyer's departure. 

\vspace{-0.2cm}
\section{Chain: A Framework for Truthful Dynamic DAs}
\label{sec:framework}

\chain\ provides a general algorithmic framework with which to construct 
truthful dynamic DAs from well-defined single-period matching rules,
such as the tr-DA rules described in the earlier section. 

Before introducing \chain\ we need a few more definitions:
Bids reported to \chain\ are {\bf active} while $t\leq
\hat{d}_i$ (for reported departure period $\hat{d}_i$), and while
the bid is unmatched and still eligible to be matched. In each period,
a single-period matching rule is used to determine whether any of
the active bids will trade and also which (if any) of the
bids that do not match will remain active in the next period.

Now we define the building blocks, well-defined single-period 
matching rules, and introduce 
the important concept of a strong no-trade predicate, which is defined
for a single-period matching rule.

\vspace{-0.1cm}
\subsection{Building Block: A Single-Period Matching Rule}

In defining a matching rule, it is helpful to adopt $b^t\in
\mathbb{R}^m_{>0}$ and $s^t\in \mathbb{R}^n_{\leq 0}$ to
denote the active bids and active asks in period $t$, where there are $m\geq
0$ and $n\geq 0$ bids and asks respectively. The bids and asks that
were active in earlier periods but are no longer active form the {\em
history} in period $t$, denoted  
$H^t\in \mathbb{R}^h$ where $h\geq 0$ is the size of the history.

A single-period matching rule (hereafter a {\em matching rule}), 
$M_{mr}=(\pi_{mr},x_{mr})$ defines an allocation rule
$\pi_{mr}(H^t,b^t,s^t,\omega)\in\{0,1\}^{(m+n)}$ and a payment rule
$x_{mr}(H^t,b^t,s^t,\omega)\in\mathbb{R}^{(m+n)}$. Here, we include
random event $\omega\in\Omega$ to allow explicitly for stochastic
matching and allocation rules. 
\begin{definition}[well-defined matching rule]
A matching rule $M_{mr}=(\pi_{mr},x_{mr})$ is well-defined
when it is strongly truthful, no-deficit, individual-rational, and strong-feasible.
\end{definition}

Here, the properties of truthfulness, no-deficit,
and individual-rationality are exactly the single-period
specializations of those defined in the previous section. For
instance, a matching rule is truthful in this sense when the
dominant strategy for an agent in a DA defined with this rule,
and in a static environment, is to bid truthfully and for all
possible random events $\omega$. Similarly for individual-rationality. No-deficit
requires that the total payments are always non-negative. 
Strong-feasibility  
requires that exactly the same number of asks are accepted as bids,
again for all random events.
\label{sec:sm}

\begin{figure}
{\small 
\begin{algorithmic}
\Function{SimpleMatch}{$H^t$,$b^t$,$s^t$}
\State matched := $\emptyset$
\State $p^t$ := mean($|H^t|$)
\While{ $(b^t\neq \emptyset) \& (s^t\neq \emptyset)$}
\State $i$ := 0, $b_i$ := $-\epsilon$, $j$ := 0, $s_j$ := $-\infty$
\While{ $(b_i<p^t) \& (b^t \neq \emptyset)$}
\State $i$ := random($b^t$), $b^t$ := $b^t \setminus \{i\}$
\EndWhile
\While{ $(s_j<-p^t) \& (s^t \neq \emptyset)$}
\State $j$ := random($s^t$), $s^t$ := $s^t \setminus \{j\}$
\EndWhile
\If{ $(i\neq 0) \& (j\neq 0)$}
\State matched := matched $\cup$ $\{ (i,j) \}$
\EndIf
\EndWhile
\EndFunction
\end{algorithmic}
}
\caption{\label{fig:sm} \small A well-defined matching rule
defined in terms of the mean bid price in the history.}
\vspace{-0.3cm}
\end{figure}

One example of a well-defined matching rule is the tr-DA, which is
invariant to the history of bids and asks.  For an
example of a well-defined, adaptive (history-dependent) and  price-based 
matching rule, consider procedure {\sc SimpleMatch} in Figure~\ref{fig:sm}.
The {\sc SimpleMatch} 
matching rule computes the mean of the absolute value of the bids
and asks in the history $H^t$ and adopts this as the clearing price in
the current period. It is a stochastic matching rule because
bids and asks are picked from the sets $b^t$ and $s^t$  at random and
offered the price.
We can reason about the properties of {\sc SimpleMatch} as follows:

(a) truthful: the price $p^t$ is independent of the bids and the
probability that a bid (or ask) is matched is independent of its
bid (or ask) price

(b) no-deficit: payment $p^t$ is collected from each matched buyer and
made to each matched seller

(c) individual-rational: only bids $b_i\geq p^t$ and asks $s_j\geq
-p^t$ are accepted.

(d) feasible: bids and asks are introduced to the ``matched" set in 
balanced pairs

\vspace{-0.1cm}
\subsection{Reasoning about Trade (Im)Possibility}

In addition to defining a matching rule $M_{mr}$, we allow a designer to
(optionally) designate a subset of losing bids that satisfy a 
property of {\em strong no-trade.} Bids that satisfy strong no-trade
are losing bids for which trade was not possible at any bid price
(c.f. ask price for asks), and moreover for which additional
independence conditions hold between bids provided with this
designation.

We first define the weaker concept of no-trade. In the following,
notation $\pi_{mr,i}(H^t,b^t,s^t,\omega|\hat{w}_i)$ indicates the
allocation decision made for bid (or ask) $i$ when its bid (ask) price
is replaced with $\hat{w}_i$:\sloppy
\begin{definition}[no-trade]
Given matching rule $M_{mr}=(\pi_{mr},x_{mr})$ then the set
of agents, $\mathrm{NT}^t$, for which no trade is possible in period
$t$ and given random events $\omega$
are those for which $\pi_{mr,i}(H^t,b^t,s^t,\omega | \hat{w}_i)=0$, for
every $\hat{w}_i\in \mathbb{R}_{> 0}$ when $i\in b^t$ and for every
$\hat{w}_i\in \mathbb{R}_{\leq 0}$ when $i\in s^t$.
\end{definition}

It can easily happen that no trade is possible, 
for instance when the agent is a buyer and there are no sellers on
the other side of the market. 
Let $\mathrm{SNT}^t\subseteq \mathrm{NT}^t$ denote the set 
of agents designated with the property of {\em strong no-trade}. Unlike
the no-trade property, strong no-trade need not be uniquely defined for a matching
rule. To be valid, however, the construction offered by a 
designer for strong no-trade must satisfy the following:
\begin{definition}[strong no-trade]
A construction for strong no-trade,
$\mathrm{SNT}^t\subseteq \mathrm{NT}^t$,
is {\bfseries valid} for a matching rule when:

(a) $\forall i\in \mathrm{NT}^t$ with $\hat{d}_i>t$, whether or not $i\in
\mathrm{SNT}^t$ is unchanged for all alternate reports
$\theta'_i=(a'_i,d'_i,w'_i)\neq \hat{\theta}_i$ while $d'_i>t$,

(b) $\forall i\in \mathrm{SNT}^t$ with $\hat{d}_i>t$, the set $\{j : j\in
\mathrm{SNT}^t, j\neq i, \hat{d}_j>t\}$ is unchanged for all reports 
$\theta'_i=(a'_i,d'_i,w'_i)\neq \hat{\theta}_i$ while $d'_i>t$, and independent
even of whether or not agent $i$ is present in the market.
\end{definition}

The strong no-trade conditions must be checked only for agents with 
a reported departure later than the current
period. Condition (a) requires that such an agent in $\mathrm{NT}^t$ 
cannot affect whether or not it
satisfies the strong no-trade predicate as long as it continues 
to report a departure later than the current period. Condition
(b) is defined recursively, and requires that if such an agent is
identified as satisfying strong no-trade, then its own report must not
affect the designation of strong no-trade to {\em other} 
agents, with reported departure later
than the current period, while it continues to report a departure
 later than the current period -- even if it delays its
reported arrival until a later period.

Strong no-trade allows for flexibility
in determining whether or not a bid is eligible for matching. 
Specifically, only
those bids that satisfy strong no-trade amongst those that lose in
the current period can remain as 
a candidate for trade in a future period. The property is defined
to ensure that such a ``surviving" agent does not, and could not,
affect the set of 
other agents against which it competes in future periods. 

\vspace{5cm}
\begin{example} 
Consider the tr-DA matching rule defined earlier with bids and asks

\vspace{1mm}
\centerline{\begin{tabular}{cc|cc}
\multicolumn{2}{c|}{\bf B} & \multicolumn{2}{c}{\bf S} \\
$i$ & $\hat w_i$ & $i$ & $\hat w_i$ \\\hline
$b_1^\ast$ & $10$ & $s_1^\ast$ & $-4$\\
$b_2$ & $8$ & $s_2$ & $-6$\\
$b_3$ & $6$ & $s_3$ & $-8$\\
\end{tabular}}
\vspace{1mm}

\noindent Bid 1 and ask 1 trade at price $8$ and $-6$
respectively. $\mathrm{NT}^t=\emptyset$ because bids 2 and 3 could
each trade if they had (unilaterally) submitted a bid price of greater
than 10. Similarly for asks 2 and 3. Now consider the order book

\vspace{1mm}
\centerline{\begin{tabular}{cc|cc}
\multicolumn{2}{c|}{\bf B} & \multicolumn{2}{c}{\bf S} \\
$i$ & $\hat w_i$ & $i$ & $\hat w_i$ \\\hline
$b_1$ & $8$ & $s_1$ & $-6$\\
$b_2$ & $7$ & $s_2$ & $-10$\\
$b_3$ & $2$ & $s_3$ & $-12$\\
\end{tabular}}
\vspace{1mm}

\noindent No trade occurs. In this case,
$\mathrm{NT}^t=\{b_1,b_2,b_3,s_1\}$. No trade is possible for any bids, even
bids 2 and 3, because $\hat w_{b_1}+\hat w_{s_2}=8-10<0$. 
But, trade is possible for
asks 2 and 3, because $\hat w_{b_2}+\hat w_{s_1}=7-6\geq 0$ and either ask 
could trade by submitting a low enough ask price.
\end{example}

\begin{example}
Consider the tr-DA matching rule and explore possible
alternative constructions for strong no-trade. 

(i) Dictatorial: in each period $t$, identify an agent that could be 
present in the period in a way that is oblivious to all agent
reports. Let $i$ denote the index of this agent. If $i\in
\mathrm{NT}^t$, then include $\mathrm{SNT}^t=\{i\}$. Strong no-trade
condition (a) is satisfied because whether or not $i$ is selected
as the ``dictator" is agent-independent, and given that it is selected,
then whether or not trade is possible is 
agent-independent. Condition (b) is trivially
satisfied because $|\mathrm{SNT}^t|=1$ and there is no cross-agent
coupling to consider. 

(ii) $\mathrm{SNT}^t:=\mathrm{NT}^t$. Consider the order book

\vspace{1mm}
\centerline{\begin{tabular}{cc|cc}
\multicolumn{2}{c|}{\bf B} & \multicolumn{2}{c}{\bf S} \\
$i$ & $\hat w_i$ & $i$ & $\hat w_i$ \\\hline
$b_1$ & $3$ & $s_1$ & $-4$\\
$b_2$ & $2$ & $s_2$ & $-6$\\
$b_3$ & $1$ & $s_3$ & $-8$\\
\end{tabular}}
\vspace{1mm}

\noindent Suppose all bids and asks remain in the market for at least 
one more period. 
Clearly, $\mathrm{NT}^t=\{b_1,b_2,b_3,s_1,s_2,s_3\}$. Consider
the candidate construction $\mathrm{SNT}^t=\mathrm{NT}^t$. Strong
no-trade condition (a) is satisfied because whether or not $i$ is in
set $\mathrm{NT}^t$ is agent-independent. Condition (b) is not
satisfied, however. Consider bid 2. If bid 2's report had been $8$ instead
of 2 then trade would be possible for bids 1 and 3, and
$\mathrm{SNT}^t=\mathrm{NT}^t=\{b_2,s_1,s_2,s_3\}$. Thus, whether or not 
bids 1 and 3 satisfy the strong no-trade predicate depends on the 
value of bid 2.  This is not a valid construction for strong no-trade
for the tr-DA matching rule.\sloppy

(iii) $\mathrm{SNT}^t=\mathrm{NT}^t$ if $|b^t|<2$ or $|s^t|<2$, and
$\mathrm{SNT}^t=\emptyset$ otherwise. As above, strong no-trade
condition (a) is immediately satisfied. Moreover, condition (b) is now
satisfied because trade is not possible for any bid or ask
irrespective of bid values because there are simply not enough bids or
asks to allow for trade with tr-DA (which needs at least 2 bids and at
least 2 asks).
\end{example}

\begin{example}
Consider a variant of the {\sc SimpleMatch} matching rule, defined
with fixed price 9. We can again ask whether
$\mathrm{SNT}^t:=\mathrm{NT}^t$ is a valid construction for strong
no-trade. Throughout this example suppose all bids and asks remain
in the market for at least one more period.
First consider a bid with $\hat w_{b_1} = 8$ and two asks with values
$\hat w_{s_1} = -6$ and $\hat w_{s_2} = -7$.
Here, $\mathrm{NT}^t=\{s_1,s_2\}$ because the asks 
cannot trade whatever their price since the bid is not high enough
to meet the fixed trading price of 9. Moreover,
$\mathrm{SNT}^t=\{s_1,s_2\}$ is a valid construction; strong no-trade
condition (a) is satisfied as above and condition (b) is satisfied
because whether or not ask 2 is in $\mathrm{NT}^t$ (and thus
$\mathrm{SNT}^t$) is independent of the price on ask 1, and vice
versa. But consider instead a bid with $\hat w_{b_1} = 8$ 
and an ask with $\hat w_{s_1} = -10$. 
Now, $\mathrm{NT}^t=\{b_1,s_1\}$ and $\mathrm{SNT}^t=\{b_1,s_1\}$
is our candidate strong no-trade set. However if bid 1 had declared value 10 instead
of 8 then $\mathrm{NT}^t=\{b_1\}$ and ask 1 drops out of
$\mathrm{SNT}^t$. Thus, strong no-trade condition (b) is not
satisfied.
\end{example}

We see from the above examples that it can be quite delicate to provide
a valid, non-trivial construction of strong no-trade. Note, however, 
that  $\mathrm{SNT}^t=\emptyset$ is a (trivial) valid
construction for any matching rule. Note also that the
strong no-trade conditions (a) and (b) require information about
the reported departure period of a bid. Thus, while the matching rules
do not use temporal information about bids, this information is
used in the construction for strong no-trade.

\vspace{-0.1cm}
\subsection{Chain: From Matching Rules to Truthful, Dynamic DAs}

The control flow in 
\chain\ is illustrated in Figure~\ref{fig:chain}. Upon arrival of a
new bid, an {\em admission decision} is made and bid $i$ is admitted
if its value $\hat{w}_i$ is at least its admission price $q_i$. 
An admitted bid competes in a sequence of {\em matching
events}, where a matching event simply applies the
 matching rule to the set of active bids and asks. If a bid fails
to match in some period and is not in the strong no-trade set
($i\notin \mathrm{SNT}^t$), then
it is {\em priced out} and leaves the market without
trading. Otherwise, if it is still before its departure time
$(t\leq \hat{d}_i),$ then it is available for matching in
the next period.
\begin{figure}
\centerline{\psfig{figure=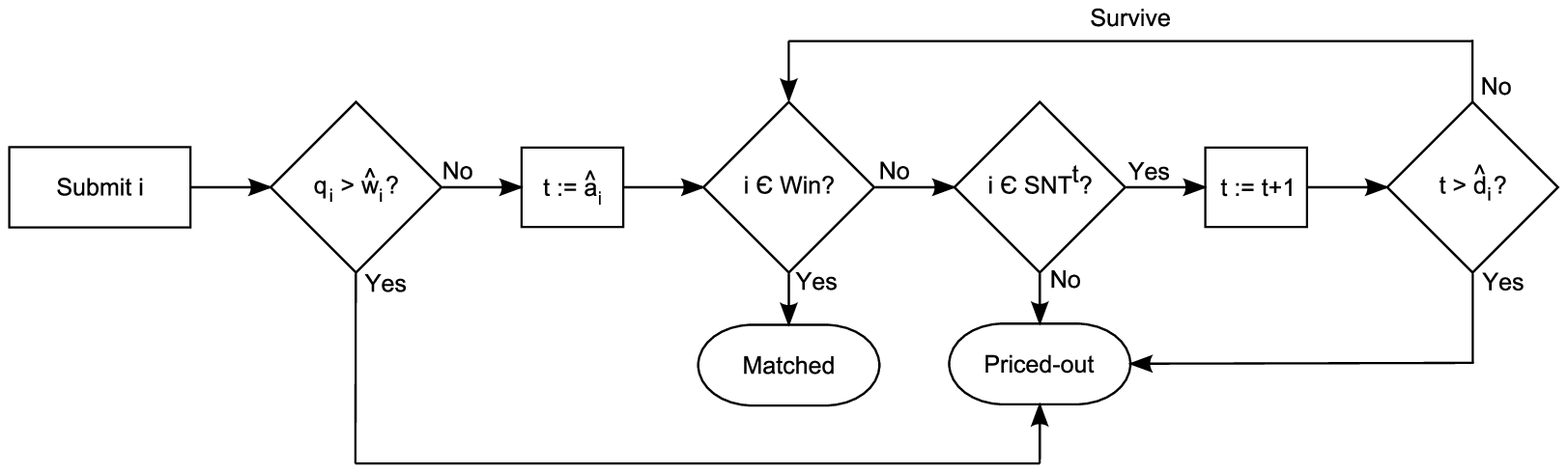,width=15.5cm}}
\caption{\small The decision process in \chain\ upon arrival 
of a new bid. If admitted, then the bid participates in a sequence of matching
events while it remains unmatched and in the strong no-trade set.
The bid matches at the first available opportunity and is priced
immediately.
\label{fig:chain}}
\vspace{-0.5cm}
\end{figure}

Each bid is always in one of three states: {\em active},
{\em matched} or {\em priced-out}. Bids are active 
if they are admitted to the market until $t\leq \hat{d}_i$, 
or until they are matched or priced-out.
An active bid becomes
{\em matched} in the first period (if any) when it trades 
in the single-period matching rule. An active bid is marked as {\em
priced-out} in the first period in which it loses but is not in
the strong no-trade set. As soon as a bid is no longer active,
it enters the history, $H^t,$ and the information about its
bid price can be used in defining matching rules for future periods.

Let $E^t$ denote the set of bids that will expire in the current
period. A well-defined matching rule, when coupled with a valid strong
no-trade construction, must provide \chain\ with the following
information, given history $H^t$, active bids $b^t$ and active
asks $s^t$, and expiration set $E^t$ in period $t$:
\vspace{0.2cm}

(a) for each bid or ask, whether it wins or loses

(b) for each winning bid or ask, the payment collected (negative
for an ask)

(c) for each losing bid or ask, whether or not it satisfies the strong
no-trade condition 
\vspace{0.2cm}

Note that the expiration set $E^t$ is only used for the strong no-trade
construction. This information is not made available to the matching rule. 
The following table summarizes the use of this information within 
\chain. Note that a winning bid cannot be
in set $\mathrm{SNT}^t$:

{
\center{\begin{tabular}{c|cc}
& $\neg \mathrm{SNT}^t$ & $\mathrm{SNT}^t$\\ \hline
Lose & priced-out & survive\\
Win & matched & n/a
\end{tabular}}

}
\vspace{0.4cm}

We describe \chain\ by defining the events that occur for a bid upon
its arrival into the market, and then in each period in which it
remains active:

\begin{itemize}
\item[] {\em Upon arrival}: Consider all possible earlier arrival periods
$t'\in[\hat{d}_i-K,\hat{a}_i-1]$ consistent with the reported type.
There are no such periods to consider if the bid is maximally patient.
If the bid would lose and not be in $\mathrm{SNT}^{t'}$
for any one of these arrival periods $t',$ then it is not admitted.
Otherwise, the bid would win in all periods $t'$ for which $i\notin
\mathrm{SNT}^{t'}$, and define the admission price as:
\begin{align}
\label{eqn:admit-price}
q(\hat{a}_i,\hat{d}_i,\theta_{-i},\omega)&:=\max_{t'\in[\hat{d}_i-K,\hat{a}_i-1], i\notin \mathrm{SNT}^{t'}}[p^{t'}_i,-\infty],
\end{align}
where $p^{t'}_i$ is the payment the agent would have made (negative
for a seller) in arrival period $t'$ (as determined by running the
myopic matching rule in that period). When the agent would lose in all earlier arrival periods
$t'$ (and so $i\in \mathrm{SNT}^{t'}$ for all $t'$), or the bid
is maximally patient, then the admission
price defaults to $-\infty$ and the bid is admitted.
\item[] {\em While active}: Consider period $t\in[\hat{a}_i,\hat{d}_i]$.
If the bid is selected to trade by the myopic matching rule, then mark
it as {\em matched} and define final payment:
\begin{align}
\label{eq:newpay}
x^t_i(\theta^{\leq
t})&=\max(q(\hat{a}_i,\hat{d}_i,\theta_{-i},\omega),p^t_i),
\end{align}
where $p^t_i$ is the
price (negative for a seller) 
determined by the myopic matching rule  in the current period. 
If this is a buyer, then collect the payment but delay
transferring the item until period $\hat{d}_i$. If this is a seller,
then collect the item
but delay making the payment until the
reported departure period. 
If the bid loses
and is not in $\mathrm{SNT}^t,$ then mark the bid as {\em priced-out}. 
\end{itemize}

We illustrate \chain\ by instantiating it to various matching
rules in the next section. 
In Section~\ref{sec:theory} we prove that \chain\ is strongly
truthful and no-deficit when coupled with a well-defined matching rule
and a valid strong no-trade construction.
We will see that the delay in buyer delivery and seller payment ensures 
truthful revelation of a trader's departure information. 
For instance, in the absence of this delay, a buyer might be able to do better 
by over-reporting departure information, still trading early enough
but now for a lower price.

\vspace{-0.1cm}
\subsection{Comments}

We choose not to allow the single-period matching rules to use the reported arrival
and departure associated with active bids and asks. This 
maintains a clean separation between non-temporal considerations
(in the  matching rules) and temporal considerations (in the
wider framework of \chain). This is also for simplicity. 
The single-period matching 
rules can be allowed to depend on the reported arrival-departure interval, as long as
the (single-period) rules are monotonic in tighter arrival-departure
intervals, in the sense that an agent that wins for some
$\hat{\theta}_i=(\hat{a}_i,\hat{d}_i,\hat{w}_i)$ continues to win and
for an improved price if it instead reports $(a'_i,d'_i,\hat{w}_i)$
with $[a'_i,d'_i]\subset [\hat{a}_i,\hat{d}_i]$. However, whether or
not trade is possible must be independent of the reported
arrival-departure interval and similarly for strong
no-trade. Determinations such as these would need to be made with
respect to the most patient type
$(\hat{d}_i-K,\hat{d}_i,\hat{w}_i)$ given report
$\hat{\theta}_i=(\hat{a}_i,\hat{d}_i,\hat{w}_i)$.

\vspace{-0.2cm}
\section{Practical Instantiations: Price-Based and Competition-Based Rules}
\label{sec:practical}

In this section we offer a number of instantiations of the 
\chain\ online DA framework. We present two different classes of well-defined
 matching rules: those that are {\em price-based} and compute simple
price statistics based on the history which are then used for
matching, and those that we refer to as {\em competition-based} and leverage
the history but also consider direct competition between the active bids and
asks in any period. In each case, we establish that the matching rules
are well-defined and provide a valid strong no-trade construction.

\vspace{-0.1cm}
\subsection{Price-Based Matching Rules}

Each one of these rules constructs a single price, $p^t$, in period
$t$ based on the history $H^t$ of earlier bids and asks
that are traded or expired.  For this purpose we define
variations on a real valued statistic, $\xi(H^t)$, that
is used to define this price given the history. Generalizing 
the {\sc SimpleMatch} procedure, as introduced in
Section~\ref{sec:sm}, the price $p^t$ is used to determine the trades in
period $t$. We also provide
a construction for {\em strong no-trade} in this context.

The main concern in setting prices is that they may be too
volatile, with price updates driving the admission price higher (via
the {\bf max} operator in the {\em admission} rule of \chain) 
and having the effect of pricing bids and asks out of the market. 
We describe various forms of smoothing and windowing, all designed 
to provide  adaptivity while dampening short-term variations. 
In each case, the parameters (e.g. the {\em smoothing factor}, or the
{\em window size}) can be determined empirically through off-line
tuning. 

We experiment with five price variants:

\paragraph{History-EWMA:} 
Exponentially-weighted moving average. The bid history,
$H^t$, is used to define price $p^t$ in period $t$, computed as
$p^t := \lambda \ \xi(H^t) + (1-\lambda)p^{t-1}$,
where $\lambda\in(0,1]$ is a
smoothing constant and $\xi(H^t)$ is a statistic defined 
for bids and asks that {\em enter} the history in period
$t$. 
Experimentally we find that the mean statistic,
$\xi^{\mathrm{mean}}(H^t),$ of
the absolute values of bids and asks that enter the history performs
well with $\lambda$ of 0.05 or lower for most scenarios that we test.
For cases in which $\xi(H^t)$ is not well-defined because of too few
(or zero) new bids or asks, then 
we set $p^t: =p^{t-1}$. 

\paragraph{History-median:} 
Compute  price $p^t$ from a statistic over
a fixed-size window of the most recent history, 
$p^t:=\xi(H^t,\Delta)$ where $\Delta$ is the window-size, i.e. defining
bids introduced to history $H^t$ in periods
$[t-\Delta,\ldots,t]$. Experimentally, we find that the median
statistic, $\xi^{\mathrm{median}}(H^t,\Delta),$ of the absolute bid and
ask values performs well for the scenarios we test, 
with the window size depending inversely with the volatility of agents'
valuations.
Typically, we observe optimal window sizes of 20 and 150, depending on 
volatility.
For cases in which $\xi(H^t,\Delta)$ is not well-defined because of 
too few (or zero) new bids or asks, then 
we set $p^t: =p^{t-1}$. 

\paragraph{History-clearing:} 
Identical to the history-median rule except the statistic 
$\xi(H^t,\Delta)$ is defined as $(b_m-s_m)/2$ where $b_m$ and $s_m$
are the lowest value pair of trades that would be executed in 
the efficient (value-maximizing) trade given all bids and asks to
enter history $H^t$ in periods
$[t-\Delta,\ldots,t]$.
Empirically, we find similar optimal window sizes for 
history-clearing as for history-median.

\paragraph{History-McAfee:} 
Define the statistic $\xi(H^t,\Delta)$ to 
represent the McAfee price, defined in Section~\ref{sec:compet}, 
for the bids in $H^t$ had they all simultaneously arrived.

\paragraph{Fixed price:} 
This simple rule computes a single fixed price
$p^t:=p^\ast$ for all trading periods, with the price optimized
offline to maximize the average-case efficiency of the dynamic DA
given \chain\ and the associated single-period matching rule that
leverages price $p^\ast$ as the candidate trading price.

\vspace{0.3cm}

For each pricing variant, procedure {\sc Match} (see
Figures~\ref{fig:matcha}--\ref{fig:matchb})  is used to determine which bids win (at price $p^t$),
which lose, and, of those that lose, which satisfy the strong no-trade
predicate. The subroutine used to determine the current price is
referred to as {\tt determineprice} in {\sc Match}. We provide
as input to {\sc Match} the set $E^t$ in addition to $(H^t,b^t,s^t)$
because {\sc Match} also constructs the strong no-trade set, and $E^t$
is used exclusively for this purpose.
\begin{figure}[p]
{
\begin{algorithmic}
\Function{Match}{$H^t$,$b^t$,$s^t$,$E^t$}
\State matched := $\emptyset$, lose := $\emptyset$, NT$^t$ := $\emptyset$,
SNT$^t$ := $\emptyset$
\State stop := false
\State $p^t$ := determineprice($H^t$)
\While{ $\neg$ stop}
\State $i$ := 0, $j$ := 0, checked$_B$ := $\emptyset$, checked$_S$ :=
$\emptyset$
\While{ ((checked$_B \subset b^t)\&$($i$=0)) $\vee$ ((checked$_S \subset
s^t)\&$($j$=0))}
\If{($i=0$)\&($j=0$)}
\State $k$ := random($b^t\setminus$ checked$_B \bigcup s^t\setminus$
checked$_S$)
\ElsIf{($i=0$)}
\State $k$ := random($b^t\setminus$ checked$_B$)
\ElsIf{($j=0$)}
\State $k$ := random($s^t\setminus$ checked$_S$)
\EndIf
\If{($k\in b^t$)}
\State checked$_B$ := checked$_B \cup \{k\}$
\If{($b_k\geq p^t$)}
\State $i$ := $k$
\EndIf
\Else
\State checked$_S$ := checked$_S \cup \{k\}$
\If{($s_k\geq -p^t$)}
\State $j$ := $k$
\EndIf
\EndIf
\EndWhile
\If{($i\neq 0$)\&($j\neq 0$)} 
\State matched := matched $\bigcup$ $\{ (i,j) \}$
\State lose := lose $\bigcup$ (checked$_B \setminus \{i\}$) $\bigcup$
(checked$_S \setminus \{j\}$)
\State $b^t$ := $b^t \setminus$ checked$_B$, $s^t$ := $s^t \setminus$
checked$_S$
\Else
\State stop := true
\EndIf
\EndWhile
\EndFunction
\end{algorithmic}
}
\caption{\label{fig:matcha} The procedure used for single-period
matching in applying \chain\ to the price-based rules. 
The algorithm continues in Figure~\ref{fig:matchb}.}
\end{figure}

\begin{figure}[t]
{
\begin{algorithmic}
\Function{Match {\rm (continued)}}{$H^t$,$b^t$,$s^t$,$E^t$}
\If{($i \neq 0$)\&($j=0$)} \Comment{I}
\State lose := lose $\bigcup s^t$, NT$^t$ := $b^t$
\If{($\exists k\in b^t\ \cdot \
((b_k\geq p^t)$\&$(\hat{d}_k=t)$)) $\vee$ ($\forall k\in s^t \ \cdot\
(\hat{d}_k=t$))}
\State SNT$^t$ := $b^t$
\Else \Comment{I-a}
\State SNT$^t$ := $b^t\setminus$ checked$_B$
\EndIf
\ElsIf{($j\neq 0$)\&$(i=0)$} \Comment{II}
\State lose := lose $\bigcup b^t$, NT$^t$ := $s^t$ 
\If{($\exists k\in s^t \ \cdot \ 
((s_k\geq -p^t)$\&$(\hat{d}_k=t)$)) $\vee$ ($\forall k\in b^t \ \cdot\
(\hat{d}_k=t$))}
\State SNT$^t$ := $s^t$
\Else
\State SNT$^t$ := $s^t\setminus$ checked$_S$
\EndIf
\ElsIf{($i=0$)\&($j=0$)} \Comment{III}
\State NT$^t$ := $b^t\bigcup s^t$
\If{($\forall k\in b^t\ \cdot\ (\hat{d}_k=t$)) $\vee$ ($\forall k\in s^t\ \cdot\
(\hat{d}_k=t$))}
\State SNT$^t$ := $b^t\bigcup s^t$
\EndIf
\EndIf
\EndFunction
\end{algorithmic}
}
\caption{\label{fig:matchb} Continuing procedure from Figure~\ref{fig:matcha}
for single-period
matching in applying \chain\ to the price-based rules.}
\end{figure}

The proof of the following lemma is technical and is postponed until
the Appendix.
\begin{lemma}
\label{lem:dp5b}
Procedure {\sc Match} defines a valid strong no-trade construction.
\end{lemma}

\begin{theorem}
Procedure {\sc Match} defines a well-defined matching rule
and a valid strong no-trade construction.
\end{theorem}
\begin{proof}~No-deficit, feasibility, and individual-rationality are immediate by
the construction of {\sc Match} since bids and asks are added to {\em
matched} in pairs, with the same payment, and only if the payment is
less than or equal to their value. Truthfulness is also easy to see:
the order with which a bid (or ask) is selected is independent of
its bid price, and the price it faces, when selected, 
is independent of its bid. If the price is less than or equal to
its bid, then whether or not it trades depends only on its order.
The rest of the claim follows from Lemma~\ref{lem:dp5b}.
\end{proof}

\begin{example}
(i) Bid $b^t=\{8\}$, ask $s^t=\{-6\}$, indexed $\{1,2\}$ 
and price $p^t=9$. The outer {\tt while} loop in Figure~\ref{fig:matcha}
terminates with $j=2$ and $i=0$ in Case II. The bid is marked as a
loser while $\mathrm{NT}^t=\{2\}$. If the bid will depart immediately,
then $\mathrm{SNT}^t=\{2\}$, otherwise $\mathrm{SNT}^t=\emptyset$.

(ii) Bid $b^t=\{8\}$, asks $s^t=\{-6,-7\}$, indexed
$\{1,2,3\}$, and price $p^t=9$. Suppose that ask 2 is selected before ask 3 
in the outer {\tt
while} loop. Then the loop terminates with $j=2$ and $i=0$ in Case II
and $\mathrm{NT}^t=\{2,3\}$. Suppose the bid and asks leave
the market later than this period. Then $\mathrm{SNT}^t=\{3\}$ because
$\mathrm{checked}_S=\{2\}$.

(iii) Bid $b^t=\{8\}$ and ask $s^t=\{-10\}$, indexed $\{1,2\}$, price
$p^t=9$ and both the bid and the ask is patient. The outer {\tt while} loop
terminates with $i=0$ and $j=0$ in Case III so that
$\mathrm{NT}^t=\{1,2\}$. However, $\mathrm{SNT}^t=\emptyset$.
\end{example}

\vspace{-0.1cm}
\subsection{Competition-Based Matching Rules}
\label{sec:compet}

Each one of these rules determines which bids match in the current
period through price competition between the active bids.
We present three variations: {\em McAfee, Windowed-McAfee and
Active-McAfee}. The latter two rules are hybrid rules in that 
they leverage history of past offers, in smoothing prices generated
by the competition-based matching rules.

\label{sec:mcafee}
\paragraph{McAfee:} Use the static DA protocol due to McAfee
as the  matching rule. Let $B$ denote the set of bids and $S$ denote
the set of asks. If $\min(|B|,|S|)<2,$ then there is no trade.
Otherwise, first insert two
dummy bids with value $\{\infty,0\}$ and two dummy asks with value
$\{0,-\infty\}$ into the set of bids and asks. Let $b_0\geq b_1\geq
\ldots\geq b_m$ and $s_0\geq s_1\geq \ldots \geq s_{n}\ldots$  
denote the bid and ask values with $(b_0,s_0)$ denoting dummy pair
$(\infty,0)$ and $(b_m,s_n)$ denoting dummy pair $(0,-\infty)$
and ties otherwise broken at random. Let $m\geq 0$ index the last pair
of bids and asks to  
clear in the efficient trade, such that $b_m+s_m\geq 0$ and
$b_{m+1}+s_{m+1}<0$. When $m\geq 1,$ consider the following two 
cases:
\begin{itemize}
\item (Case I) If price
$p_{m+1}=\frac{b_{m+1}-s_{m+1}}{2}\leq b_m$ and $-p_{m+1}\leq s_m$ then
the first $m$ bids and asks trade and payment $p_{m+1}$ is collected
from each winning buyer and made to each winning seller.
\item (Case II) Otherwise, the first $m-1$ bids and asks trade and 
payment $b_m$ is collected from
each winning buyer and payment $-s_m$ is made to each winning seller.  
\end{itemize}

To define $\mathrm{NT}^t$, replace a bid that does not trade with 
a bid reporting a very large value and see whether this bid
trades. To determine whether trade is possible for an ask that does not
trade: replace the ask with an ask reporting value $\epsilon>0$, some small
$\epsilon$. 
Say that there is a {\em quorum} if and only if there are at least
two bids and at least two asks, i.e. $\min(|b^t|,|s^t|)\geq 2$.
Define {\bf strong no-trade} as follows:
set $\mathrm{SNT}^t:=\mathrm{NT}^t=b^t\cup s^t$ when there
is no quorum and $\mathrm{SNT}^t:=\emptyset$ otherwise.
\begin{lemma}
\label{lem:dp8}
For any bid $b_i$ in the McAfee matching rule, then for any other
bid (or ask) $j$ there is some bid $\hat{b}_i$ that will make
trade possible for bid (or ask) $j$ when there is a quorum.
\end{lemma}
\begin{proof}~Without loss of generality, suppose there are three bids and three
asks. Label the bids $(a,c,e)$ and the asks $(b,d,f)$, both ordered 
from highest to lowest so that $(a,b)$ is the most competitive bid-ask
pair. Proceed by case analysis on bids. The
analysis is symmetric for asks and omitted. Let
$\mathit{tp}(i)\in\{0,1\}$ denote whether or not trade is possible for
bid $i$, so that $i\in\mathrm{NT}^t \ \Leftrightarrow \ \mathit{tp}(i)=0$.
For bid $a$: when $b\geq
-(a-d)/2$ then $\mathit{tp}(c)=\mathit{tp}(e)=1$ and this inequality
can always be satisfied for a large enough a; when 
$a\geq (c-d)/2$ then $\mathit{tp}(b)=1$ and when $a\geq (c-b)/2$ then
$\mathit{tp}(d)=\mathit{tp}(f)=1$, and both of these inequalities
are satisfied for a large enough $a$. For bid $c$: when $b\geq
-(c-d)/2$ then $\mathit{tp}(a)=1$ and when, in addition, $c>a$, then
$\mathit{tp}(e)=1$ and each one of these inequalities are satisfied for a large
enough $c$; similarly when $c\geq (a-d)/2$ then $\mathit{tp}(b)=1$ and when $c\geq
(a-b)/2$ then $\mathit{tp}(d)=\mathit{tp}(f)=1$. Analysis for bid $e$
follows from that for bid $c$. 
\end{proof}
\begin{lemma}
The construction for strong no-trade is valid and 
there is no valid strong no-trade construction that 
includes more than one losing bid or ask that will not
depart in the current period for any period in which
there is a quorum.
\end{lemma}
\begin{proof}~To see that this is a valid construction, 
notice that strong no-trade
condition (a) holds since any bid (or ask) is always in both
$\mathrm{NT}^t$ and $\mathrm{SNT}^t$. Similarly, condition (b) 
trivially holds (with the other bids and asks remaining in
$\mathrm{SNT}^t$ even if any bid is not present in the market).
To see that this definition is essentially maximal, consider now
that $\min(|b^t|,|s^t|)\geq 2$. For contradiction, suppose that 
two losing bids $\{i,j\}$ with departure after the current period are designated
as strong no-trade. But, strong no-trade condition (b) fails because
of Lemma~\ref{lem:dp8} because either bid could have submitted an
alternate bid price that would remove the other bid from
$\mathrm{NT}^t$ and thus necessarily also from $\mathrm{SNT}^t$.
\end{proof}

The construction offered for $\mathrm{SNT}^t$ 
cannot be extended even to include one
agent selected at random from the set $i\in\mathrm{NT}^t$ that will not 
depart immediately, in the case of a quorum. Such a construction would
fail strong no-trade condition (b) when the set $\mathrm{NT}^t$
contains more than one bid (or ask) that does not depart in the
current period, because bid $i$'s absence from the market would cause
some other agent to be (randomly) selected as $\mathrm{SNT}^t$.
\vspace{0.3cm}

\paragraph{Windowed-McAfee:} This myopic matching rule is
parameterized on window size $\Delta$.
Augment the active bids and asks
with the bids and asks introduced to the history $H^t$ 
in periods $t'\in \{t-\Delta+1,\ldots,t\}$.
Run {\em McAfee} with this augmented set of bids and asks and determine which
of these bids and asks would trade. Denote this candidate set $C$. 
Some active agents identified as matching in $C$ 
may not be able to trade in this period because $C$ can also contain
non-active agents.

Let $B'$ and $S'$ denote, respectively, the {\em active}
bids and {\em active} asks in set $C$. Windowed-McAfee then proceeds
by picking a random subset of $\min(|B'|,|S'|)$ bids and asks
to trade. When $|B'|\neq |S'|,$ then
some bids and asks will not trade.
\vspace{0.3cm}

Define {\bf strong no-trade} for this matching rule as:
\vspace{0.2cm}

(i) if there are no active asks but active bids, then
$\mathrm{SNT}^t:=b^t$

(ii) if there are no active bids but active asks, then
$\mathrm{SNT}^t:=s^t$

(iii) if there are fewer than 2 asks or fewer than 2 bids in the {\em augmented}
bid set, then $\mathrm{SNT}^t:=b^t\cup s^t$,
\vspace{0.2cm}

\noindent and otherwise set $\mathrm{SNT}^t:=\emptyset$.
In all cases it should be clear that $\mathrm{SNT}^t\subseteq
\mathrm{NT}^t$. 

\begin{lemma}
The strong no-trade construction for windowed-McAfee is valid.
\end{lemma}
\begin{proof}~That this is a valid SNT criteria in case (iii) follows immediately
from the validity of the SNT criteria for the standard McAfee matching
rule. Consider case (i). Case (ii) is symmetric and omitted. 
For strong no-trade condition (a), we see that all bids
$i\in\mathrm{NT}^t$ and also $i\in\mathrm{SNT}^t$, and 
whether or not they are designated strong no-trade is independent of
their own bid price but simply because there are no active
asks. Similarly,  for strong no-trade condition (b), we see that all
bids (and never any asks) are in $\mathrm{SNT}^t$ whatever the bid
price of any particular bid (and even whether or not it is present).
\end{proof}

Empirically, we find that the efficiency 
of Windowed-McAfee is sensitive to the size of $H^t$, but that frequently
the best choice is a small window size that includes only the active bids.
\vspace{0.5cm}

\paragraph{Active-McAfee:}   
Active-McAfee augments the active bids and asks to include all
unexpired but traded or priced-out offers.
It proceeds as in Windowed-McAfee given this augmented bid set.

\vspace{-0.1cm}
\subsection{Extended Examples}

We next provide two stylized examples to demonstrate the matching performed 
by \chain\ using both a price-based and a competition-based matching rule.
For both examples, we assume a maximal patience of $K=2$. Moreover,
while we describe when \chain\ determines that a bid or an ask
trades, remember that a winning buyer is not allocated the good until
its reported departure and a winning seller does not receive payment
until its reported departure.

\begin{example}
Consider \chain\ using an adaptive, price-based matching rule. The particular
details of how prices are determined are not relevant. Assume that the prices
in periods 1 and 2 are $\{p^1,p^2\}=\{8,7\}$ and the maximal patience is three
periods. Now consider period 3 and suppose
that the order book is empty at the end of period 2 and that the bids and asks
in Table~\ref{tab:ex1} arrive in period 3.
\begin{table}[h!]
\setlength{\tabcolsep}{1.6mm}
\centerline{\begin{tabular}{ccccccc|ccccccc}
\multicolumn{7}{c|}{\bf B} & \multicolumn{7}{c}{\bf S} \\
$i$ & $\hat w_i$ & $\hat d_i$ & $\hat d_i\!-\!K$ & $q_i$ & $p_i$ & SNT? & 
$i$ & $\hat w_i$ & $\hat d_i$ & $\hat d_i\!-\!K$ & $q_i$ & $p_i$ & SNT?
\\ \hline
$b_1$* & 15 & 4 & 2 & 7 & 7 & N  & $s_1$ & -1 & 4 & 2 & -7 & n/a & Y
\\
$b_2$* & 10 & 3 & 1 & 8 & 8 & N & $s_2$* & -3 & 5 & 3 & $-\infty$ & -6.5 & N 
\\
\sout{\hspace{1mm}$b_3$\hspace{1mm}} & 7 & 3 & 1 & 8 & n/a & N & 
$s_3$ & -4 & 3 & 1 & -7 & n/a & Y \\
$b_4$ & 6 & 5 & 3 & $-\infty$ & n/a & N &
$s_4$* & -5 & 4 & 2 & -7 & -6.5 & N\\
& & & & & & &
$s_5$ & -10 & 5 & 3 & $-\infty$ & n/a & Y\\
\end{tabular}} 
\caption{\label{tab:ex1} \small Bids and asks that arrive in period 3. Bids $\{b_1, b_2\}$ match
with asks $\{s_2, s_4\}$ (as indicated with a *). Bid $b_3$ is
priced-out upon admission because $q_{b_3} > \hat w_{b_3}$ (indicated with a strike-through). The admission price is
$q_i$ and the payment made by an agent that trades is $p_i$. Column
`SNT?' indicates whether or not the bid or ask satisfies the strong
no-trade predicate. Asks $\{s_1, s_5\}$ survive into the next period
because they are in SNT and have $d_i>3$.}
\vspace{-0.3cm}
\end{table}

Bids $\{b_1, b_2, b_4\}$ and asks $\{s_1,..,s_5\}$ are admitted. Bid
$b_3$ is priced out because $q_{b_3} = \max(p^1, p^2,
-\infty)=\max(8,7,-\infty)=8 > \hat w_{b_3} = 7$ by
Eq.~(\ref{eqn:admit-price}). Note that $b_4$ and $s_5$ are
admitted despite low bids (asks) because they have maximal patience
and their admission prices are $-\infty$. Now, suppose that $p^3:=6.5$ is
defined by the matching rule and consider applying {\sc
Match} to the admitted bids and asks. 

Suppose that the bids are randomly ordered as $(b_4, b_2, b_1)$ and
the asks as $(s_4, s_2, s_1, s_3, s_5)$. Bid $b_4$ is picked first but
priced-out because $\hat w_{b_4}=6 < p^3=6.5$. Bid $b_2$ is tentatively accepted
($\hat w_{b_2}=10\geq p^3=6.5$) and then ask $s_4$ is accepted ($w_{s_4}=-5\geq
p^3=-6.5$). Bid $b_2$ is matched with ask $s_4$, with payment
$\max(q_{b_2},p^3)=\max(8,6.5)=8$ for $b_2$ by
Eq.~(\ref{eq:newpay}) and payment $\max(q_{s_4},p^3)=\max(-\infty,-6.5)=-6.5$
for $s_4$. Bid $b_1$ is then tentatively accepted ($15\geq 6.5$) and
then matched with ask $s_2$, which is accepted because $-3\geq
-6.5$. The payments are $\max(7,6.5)=7$ for $b_1$ and
$\max(-\infty,-6.5)=-6.5$ for $s_2$. 
Ask $s_3$ expires but asks $s_1$ and $s_5$ survive and are marked
$i\in \mathrm{SNT}$ in this period because they were never offered the
chance to match with any bid. These asks will be active in period 4. 

Note the role of the admission price in truthfulness. Had bid $b_1$
delayed arrival until period 4, its admission price would be
$\max(p^2,p^3,-\infty)=\max(7,6.5)=7$ and its payment in period 4 (if it
matches) at least 7. Similarly, had ask $s_4$ delayed arrival, then its
admission price would be $\max(-7,-6.5,-\infty)=-6.5$ and the maximal
payment it can receive in period 4 is 6.5.

\end{example} 

\begin{example}
Consider \chain\ using the McAfee-based matching rule with $K=3$ and with the same bids and asks arriving in
period 3. Suppose that the prices in periods 1 and 2 that would have
been faced by a buyer are $\{p^1_b, p^2_b\}=\{8,7\}$ and
$\{p^1_s,p^2_s\}=\{-7,-6\}$ for a seller. These prices are determined
by inserting an additional bid (with value $\infty$) or an additional
ask (with value 0) into the order books in each of periods 1 and 2. We
will illustrate this for period 3. Consider now the bids and asks
in period 3 in Table~\ref{tab:ex2}.
\begin{table}[h!]
\setlength{\tabcolsep}{1.6mm}
\centerline{\begin{tabular}{ccccccc|ccccccc}
\multicolumn{7}{c|}{\bf B} & \multicolumn{7}{c}{\bf S} \\
$i$ & $w_i$ & $d_i$ & $d_i\!-\!K$ & $q_i$ & $p_i$ & SNT? & 
$i$ & $w_i$ & $d_i$ & $d_i\!-\!K$ & $q_i$ & $p_i$ & SNT?
\\ \hline
$b_1$* & 15 & 4 & 2 & 7 & 7 & N  & $s_1$* & -1 & 4 & 2 & -6 & -4 & N
\\
$b_2$* & 10 & 3 & 1 & 8 & 8 & N & $s_2$* & -3 & 5 & 3 & $-\infty$ & -4 & N 
\\
\sout{\hspace{1mm}$b_3$\hspace{1mm}} & 7 & 3 & 1 & 8 & n/a & N & 
$s_3$ & -4 & 3 & 1 & -6 & n/a & N \\
$b_4$ & 6 & 5 & 3 & $-\infty$ & n/a & N &
$s_4$ & -5 & 4 & 2 & -6 & n/a & N\\
& & & & & & &
$s_5$ & -10 & 5 & 3 & $-\infty$ & n/a & N\\
\end{tabular}} 
\caption{\label{tab:ex2} \small Bids and asks that arrive in period 3. Bids $\{b_1, b_2\}$ match
with asks $\{s_1, s_2\}$ (as indicated with a *). Bid $b_3$ is
priced-out upon admission because $q_{b_3} > \hat w_{b_3}$. 
The admission price is
$q_i$ and the payment made by an agent that trades is $p_i$. Column
`SNT?' indicates whether or not the bid or ask satisfies the strong
no-trade predicate. No asks or bids survive into the next period.}
\vspace{-0.3cm}
\end{table}

As before bid $b_3$ is not admitted. The myopic matching rule now runs
the (static) McAfee auction rule on bids $\{b_1,b_2, b_4\}$ and asks
$\{s_1,..,s_5\}$. Consider bids and asks in decreasing order of value,
the last efficient trade is indexed $m=3$ with $\hat w_{b_4}+\hat w_{s_3}=6-4\geq 0$. 
But $p_{m+1}=(0-(-5))/2=2.5$ (inserting a dummy bid with value 0 as
described in Section~\ref{sec:mcafee}). Price $-p_{m+1}=-2.5>s_3=-4$
and this trade cannot be executed by McAfee. Instead, buyers
$\{b_1,b_2\}$ trade and face price $p^b_m=\hat w_{b_4}=6$ and sellers
$\{s_1,s_2\}$ trade and face price $p^s_m=\hat w_{s_3}=-4$. Bids $b_4$ and asks
$\{s_3,s_4,s_5\}$ are priced-out and do not survive into the next
round. Ultimately, payment $\max(q_{b_1},p^b_m)=\max(7,6)=7$ is collected
from buyer $b_1$ and payment $\max(q_{b_2},p^b_m)=\max(8,6)=8$ is
collected from buyer $b_2$. For sellers, payment $\max(-6,-4)=-4$ and
$\max(-\infty,-4)=-4$ for $s_1$ and $s_2$ respectively. 

The prices $p^3_b$ and $p^3_s$ that are used in 
Eq.~(\ref{eqn:admit-price}) to define the admission price for
bids and asks with arrivals in periods 4 and 5 are determined
as follows. For the buy-side price, we introduce an additional bid with
bid-price $\infty$. With this the bid values considered by McAfee would be
$(\infty,15,10,6,0)$ and the ask values would be $(-1,-3,-4,-5,-10)$, where
a dummy bid with value 0 is included on the buy-side. The last
efficient pair to trade is $m=4$ with $6-5\geq 0$ and
$p_{m+1}=(0-(-10))/2=5$, which satisfies this bid-ask pair. Therefore
the buy-side price, $p^3_b:=5$. On the sell-side, we introduce an
additional ask with ask-price $0$ so that the bid values considered by
McAfee are $(15,10,6,0)$ (again, with a dummy bid included) and the
ask values are $(0,-1,-3,-4,-5,-10)$. This time $m=3$ and the last efficient
pair to trade is $6-3\geq 0$. Now $p_{m+1}=(0-(-4))/2=2$ and
this price does not satisfy $s_2$, with $-p_{m+1}>s_2$ and price
$p^s_{m+1}=s_2=-3$ is adopted. Therefore the sell-side price,
$p^3_s:=-3$. 

Again, we can see that bidder 1 cannot improve its price by delaying
its entry until period 4. The admission price for the bidder would be
$\max(p^2_b,p^3_b)=\max(7,p^3_b)\geq 7$ and thus its payment in period
4, if it matches, will be at least 7.\footnote{We can check that
$p^3_b:=6$ in this case. Suppose that bidder 1 were not present in
period 3. Now
consider introducing an additional bid with value $\infty$ so that the bids 
values are $\{\infty,10,6,0\}$ (with a dummy bid) with ask values
$\{-1,-3,-4,-5,-10\}$. Then $m=3$ and $p_{m+1}=(0-(-5))/2=2.5$, which
does not support the trade between bid $b_4$ and ask
$s_3$. Instead, $p^b_m=\hat w_{b_4}=6$ is adopted, and we would have
$p^3_b:=6$. Of course, this is exactly the price determined by McAfee
for bid $b_1$ in period 3 when the bidder is truthful.} Similarly for
ask $s_1$, which would face admission price
$\max\{p^2_s,p^3_s\}=\max\{-6,-4\}=-4$ and can receive a payment of at
most 4 in period 4. We leave it as an exercise for the
reader to verify that $p^3_s=-4$ if ask $s_1$ delays its
arrival until period 4 (in comparison, $p^3_s=-3$ when ask $s_1$ is
truthful). 

\end{example}

Because the McAfee-based pricing scheme computes a price and clears the order
book following every period in which there are at least two bids and two asks,
the bid activity periods tend to be short in comparison to the
adaptive, price-based rules where orders can be kept active longer when
there is an asymmetry in the number of bids and asks in the market. 
In fact, one interesting artifact that occurs with adaptive, price-based
matching rules is that
the admission-price and $\mathrm{SNT}$ can perpetuate this kind of
bid-ask asymmetry.
Once the market has more asks than bids, $\mathrm{SNT}$ becomes likely for future
asks, but not bids.
Therefore, bids are much more likely than asks 
to be immediately priced out of the market
by failing to meet the admission price constraint.

\vspace{-0.2cm}
\section{Theoretical Analysis: Truthfulness, Uniqueness, and Justifying Bounded-Patience}
\label{sec:theory}

In this section we prove that \chain\ combined with a well-defined
matching rule and a valid strong no-trade construction generates a
truthful, no-deficit, feasible and individual-rational dynamic DA.
In Section~\ref{sec:nec}, we establish that  uniqueness of 
\chain\ amongst dynamic DAs that are constructed from single-period
matching rules as building blocks. In Section~\ref{sec:neg}, we 
establish the importance of the existence of a maximal bound on bidder patience
by presenting a simple 
environment in which no truthful, no-deficit DA can implement
even a single trade despite the number of efficient trades
can be increased without bound. 

\vspace{-0.1cm}
\subsection{The Chain Mechanism is Strongly Truthful}
\label{sec:truth}

It will be helpful to adopt a price-based interpretation of a valid
single-period matching rule. Given rule $M_{mr}$, define an
{\em agent-independent price}, $z_i(H^t,A^t\setminus i,\omega)\in
\mathbb{R}$ where $A^t=b^t\cup s^t$, such that for all $i$, all bids $b^t$,
all asks $s^t$, all history $H^t$, and all random events
$\omega\in\Omega$.
We have:
\begin{enumerate}
\item[] (A1) $\hat{w}_i-z_i(H^t,A^t\setminus i,\omega)>0\Rightarrow
\pi_{mr,i}(H^t,b^t,s^t,\omega)=1$, and $\hat{w}_i-z_i(H^t,A^t\setminus i,\omega)<0\Rightarrow
\pi_{mr,i}(H^t,b^t,s^t,\omega)=0$
\item[] (A2) payment $x_{mr,i}(H^t,b^t,s^t,\omega)=z_i(H^t,A^t\setminus
i,\omega)$ if $\pi_{mr,i}(H^t,b^t,s^t,\omega)=1$ and
$x_{mr,i}(H^t,b^t,s^t,\omega)=0$ otherwise
\end{enumerate}

The interpretation is that there is an agent-independent price,
$z_i(H^t,A^t\setminus i,\omega)$, that is at least $\hat{w}_i$ when
the  agent loses and  no greater than $\hat{w}_i$ otherwise.
In particular, $z_i(H^t,A^t\setminus i,\omega)=\infty$ when
$i\in\mathrm{NT}^t$. Although an agent's price is only 
explicit in a matching rule when the agent trades, it is well known
that such a price exists for any truthful, single-parameter mechanism;
e.g., see works by Archer and Tardos~\citeyear{archer01} and
Goldberg and Hartline~\citeyear{goldberg03a}.\footnote{A single-parameter mechanism is one
in which the private information of an agent is limited to one
number. This fits the single-period matching problem because
the arrival and departure information is discarded. Moreover, although there are both buyers and
sellers, the problem is effectively single-parameter because no buyer
can usefully pretend to be a seller and {\em vice versa}.} Moving
forward we adopt price $z_i$ to characterize the matching rule used as
a building block for \chain, and assume without loss of
generality properties (A1) and (A2). 

Given this, we will now establish the truthfulness of \chain\ by
appeal to a price-based characterization
due to Hajiaghayi \etal~\citeyear{hajiaghayi05} for truthful, dynamic
mechanisms. We state (without proof) a variant on the characterization result 
that holds for stochastic policies $(\pi,x)$ and {\em
strong}-truthfulness. The theorem that we state is also
specialized to our DA environment.
We continue to adopt 
$\omega\in\Omega$ to capture the realization of stochastic events internal
to the mechanism:

\begin{theorem}
\shortcite{hajiaghayi05}
\label{thm:truth}
A dynamic DA $M=(\pi,x)$, perhaps stochastic, is strongly truthful for
misreports limited to no early-arrivals if and only if, for every
agent $i$, all $\hat{\theta}_i$, all $\theta_{-i}$, and all random events $\omega\in\Omega$, 
there exists a price $p_i(\hat{a}_i,\hat{d}_i,\theta_{-i},\omega)$ such that:
\begin{enumerate}
\item[] (B1) the price is independent of agent $i$'s reported value
\item[] (B2) the price is monotonic-increasing in tighter
  $[a'_i,d'_i]\subset [\hat{a}_i,\hat{d}_i]$
\item[] (B3) trade $\pi_i(\hat{\theta}_i,\theta_{-i})=1$ whenever
  $p_i(\hat{a}_i,\hat{d}_i,\theta_{-i},\omega)< \hat{w}_i$ and $\pi_i(\hat{\theta}_i,\theta_{-i})=0$
whenever $p_i(\hat{a}_i,\hat{d}_i,\theta_{-i},\omega)>\hat{w}_i$, and the trade 
is performed for a buyer upon its departure period $\hat{d}_i$.
\item[] (B4) the agent's payment is 
  $x_i(\hat{\theta}_i,\theta_{-i})=p_i(\hat{a}_i,\hat{d}_i,\theta_{-i},\omega)$ when
  $\pi_i(\hat{\theta}_i,\theta_{-i})=1$, 
with $x_i(\hat{\theta}_i,\theta_{-i})=0$ 
otherwise, and the payment is made to a seller
upon its departure, $\hat{d}_i$.
\end{enumerate}
where random event $\omega$ is independent of
the report of agent $i$ in as much as it affects the price to agent $i$.
\end{theorem}

Just as for the single-period, price-based characterization, the price
$p_i(a_i,d_i,\theta_{-i},\omega)$ need not always be explicit in 
\chain. Rather, the theorem states that given any truthful dynamic DA,
such as \chain, there exists a well-defined price function with
these properties of value-independence (B1) and arrival-departure
monotonicity (B2), and such that they define the trade (B3) and the
payment (B4). 

To establish the truthfulness of \chain, we prove that it is
well-defined with respect to the following price function:
\begin{align}
p_i(\hat{a}_i,\hat{d}_i,\theta_{-i},\omega)&=\max(\check{q}(\hat{a}_i,\hat{d}_i,\theta_{-i},\omega),
\check{p}_i(\hat{a}_i,\hat{d}_i,\theta_{-i},\omega)),
\end{align}
where
\begin{align}
\check{q}(\hat{a}_i,\hat{d}_i,\theta_{-i},\omega)&=\max_{t\in[\hat{d}_i-K,\hat{a}_i-1],
i\notin\mathrm{SNT}^t}(z_i(H^t,A^t\setminus
i,\omega),-\infty)
\end{align}
and
\begin{align}
\check{p}(\hat{a}_i,\hat{d}_i,\theta_{-i},\omega)&=
\left\{
\begin{array}{ll}
z_i(H^{t*},A^{t*}\setminus i,\omega) & \mbox{, if
$\mathit{decision}(i)=1$}\\
+\infty & \mbox{, otherwise}
\end{array}
\right.
\end{align}
where $\mathit{decision}(i)=0$ indicates that $i\in\mathrm{SNT}^t$ for
all $t\in [\hat{a}_i,\hat{d}_i]$ and $\mathit{decision}(i)=1$
otherwise, and where $t*\in[\hat{a}_i,\hat{d}_i]$ is the first period
in which $i\notin \mathrm{SNT}^t$. We refer to this as the {\em
decision period}. Term
$\check{q}(\hat{a}_i,\hat{d}_i,\theta_{-i},\omega)$ denotes the
admission price, and is defined on periods $t$ 
before the agent arrives for which $i\notin \mathrm{SNT}^t$ had it
arrived in that period. Note carefully that the rules of \chain\
are implicit in defining this price function. For instance, 
whether or not $i\in\mathrm{SNT}^t$ in some period $t$ depends, for example, on the
other bids that remain active in that period.

We now establish conditions (B1)--(B4). The proofs
of the technical lemmas are deferred until the Appendix. 
The following lemma is
helpful and gets to the heart of the strong no-trade concept.

\begin{lemma} 
\label{lem:new}
The set of active agents (other than $i$) in period $t$ in \chain\
is independent of $i$'s report while agent $i$ remains active, and
would be unchanged if $i$'s arrival is later than period $t$.
\end{lemma}

The following result establishes properties (B1) and (B2).
\begin{lemma}
\label{lem:newdp88}
The price constructed from admission price $\check{q}$ and
post-arrival price $\check{p}$ is value-independent and monotonic-increasing when
the matching rule in \chain\ is well-defined, the strong no-trade
construction is valid, and agent patience is bounded by $K$.
\end{lemma}

Having established properties (B1) and (B2) for
price function $p_i(\hat{a}_i,\hat{d}_i,\theta_{-i},\omega)$, we just
need to establish (B3) and (B4) to show truthfulness. The timing
aspect of (B3) and (B4), which requires that the buyer receives an item and the seller receives 
its payment upon reported departure, is already 
clear from the definition of
\chain.
\begin{theorem}
The online DA {\sc Chain} is strongly truthful, no-deficit, feasible
and individual-rational when the matching rule is well-defined, the
strong no-trade construction is valid, and agent patience is bounded
by $K$.
\end{theorem}
\begin{proof}~Properties (B1) and (B2) follow from Lemma~\ref{lem:newdp88}. The 
timing aspects of (B3) and (B4) are immediate. To complete the proof,
we first consider (B3). If
$\check{q}>\hat{w}_i,$ then agent $i$ is priced out at admission
by \chain\ because this reflects that
$z_i(H^t,A^t\setminus i,\omega)>\hat{w}_i$ in some
$t\in[\hat{d}_i-K,\hat{a}_i-1]$ with $i\notin\mathrm{SNT}^t$, and thus
the bid would lose if it arrived in that period (either because it
could trade, but for a payment greater than its reported value, or
because $i\in\mathrm{NT}^t$). Also, if there is no decision period,
then $\check{p}=\infty$, which is consistent with \chain, because
there is no bid price at which a bid will trade when 
$i\in\mathrm{SNT}^t$ for all
periods $t\in[\hat{a}_i,\hat{d}_i]$. Suppose now that there is a decision
period $t*$ and $\check{q}<\hat{w}_i$. If $\check{p}>\hat{w}_i,$
then there should be no trade. This is the case in \chain, because
the price $z_i(H^{t*},A^{t*}\setminus i,\omega)$ in $t*$ is greater
than $\hat{w}_i$ and thus the agent is
priced-out. If $\check{p}<\hat{w}_i$ then the bid should trade and
indeed it does, again because the price $z_i$ in that period 
satisfies (A1) and (A2) with respect to the matching rule.
Turning to (B4), it is immediate that the payments collected in 
\chain\ are equal
to price $p_i(\hat{a}_i,\hat{d}_i,\theta_{-i},\omega)$, because if bid
$i$ trades then $p_i(\hat{a}_i,\hat{d}_i,\theta_{-i},\omega)\leq
\hat{w}_i$ and thus $\check{q}\leq \hat{w}_i$ and $\check{p}\leq
\hat{w}_i$. The admission price
$q(\hat{a}_i,\hat{d}_i,\theta_{-i},\omega)=\check{q}(\hat{a}_i,\hat{d}_i,\theta_{-i},\omega)$
when $\check{q}\leq \hat{w}_i$ because price $z_i$ is well-defined 
by properties (A1) and (A2). Similarly, the payment $p^{t*}$ defined
by the matching rule in \chain\ in the decision period is equal to
$\check{p}$. 

That \chain\ is individual-rational and feasible
follows from inspection.
\chain\ is no-deficit because the payment collected from every agent (whether a
buyer or a seller) is at least that defined by a valid matching rule
in the decision period $t*$ (it can be higher when the admission price
is higher than this matching price),
the matching rules are themselves no-deficit, and because
the auctioneer delays making a payment to a seller until its reported
departure but collects payment from a buyer immediately upon a match.
\end{proof}

\label{sec:report}
We remark that information can be reported to bidders that are not
currently participating in the market, for instance to assist in their
valuation process. If this information is 
delayed by at least the maximal patience of a
bidder, so that the bid of a current bidder cannot influence the other
bids and asks that it faces, then this is without any strategic
consequences. Of course, without this constraint, or with bidders that
participate in the market multiple times, the effect of such feedback
would require careful analysis and bring us outside of the private
values framework.

\vspace{-0.1cm}
\subsection{Chain is Unique amongst Dynamic DAs that are constructed
from Myopic Matching Rules}
\label{sec:nec}

In what follows, we establish that \chain\ is unique amongst all
truthful, dynamic DAs that adopt well-defined, myopic matching rules
as simple building blocks. For this, we define the class of 
{\em canonical, dynamic DAs}, which take a well-defined single
period matching rule coupled with a valid strong no-trade
construction, and satisfy the following requirements:
\vspace{0.2cm}

(i) agents are active until they are matched or priced-out,

(ii) agents participate in the single-period matching rule
while active

(iii) agents are matched if and only if they trade in the
single-period matching rule.
\vspace{0.2cm}

We think that these restrictions capture the essence of what it means to
construct a dynamic DA from single-period matching rules. Notice
that a number of design elements are left undefined, including the
payment collected from matched bids, when to mark an active bid as
priced-out, what rule to use upon admission, and how to use the strong
no-trade information within the dynamic DA.
In establishing a uniqueness result, we 
leverage the necessary and sufficient price-based characterization
in Theorem~\ref{thm:truth}, and exactly determine the price function
$p_i(\hat{a}_i,\hat{d}_i,\theta_{-i},\omega)$ to that defined in
Eq.~(\ref{eq:newpay}) and associated with \chain. The proofs
for the two technical lemmas are deferred until the Appendix. 

\begin{lemma}
\label{lem:dp5}
A strongly truthful, 
canonical dynamic DA must define price
$p_i(\hat{a}_i,\hat{d}_i,\theta_{-i},\omega)\geq
z_i(H^{t*},A^{t*}\setminus i,\omega)$ where $t*$ is the decision
period for bid $i$ (if it exists). Moreover, the bid must be
priced-out in period $t*$ if it is not matched.
\end{lemma}

\begin{lemma}
\label{lem:dp5a}
A strongly truthful, canonical and individual-rational  dynamic 
DA must define price $p_i(\hat{a}_i,\hat{d}_i,\theta_{-i},\omega)\geq
\check{q}(\hat{a}_i,\hat{d}_i,\theta_{-i},\omega)$, and a bid with
$\hat{w}_i<\check{q}(\hat{a}_i,\hat{d}_i,\theta_{-i},\omega)$ must be
priced-out upon admission.
\end{lemma}

\begin{theorem}
The dynamic DA algorithm \chain\ uniquely defines a strongly
truthful, individual-rational auction among canonical dynamic DAs that
only designate bids as priced-out when necessary.
\end{theorem}
\begin{proof}~If there is no decision period, then we must have
$p_i(\hat{a}_i,\hat{d}_i,\theta_{-i},\omega)=\infty$, by canonical
(iii) coupled with (B3). Combining this with 
Lemmas~\ref{lem:dp5} and~\ref{lem:dp5a}, we have $p_i(\hat{a}_i,\hat{d}_i,\theta_{-i},\omega)\geq
\max(\check{q}(\hat{a}_i,\hat{d}_i,\theta_{-i},\omega),\check{p}(\hat{a}_i,\hat{d}_i,\theta_{-i},\omega))$.
We have also established that a bid must be priced-out if its bid
value is less than the admission price, or it fails to match in its
decision period. Left to show is that the price is exactly as in 
\chain, and that a bid is admitted when its value $\hat{w}_i\geq
\check{q}(\hat{a}_i,\hat{d}_i,\theta_{-i},\omega)$ and retained as
active when it is in the strong no-trade set. The last two control
aspects are determined once we choose a rule that ``only designates
bids as priced-out when necessary." We prefer to allow a bid to remain
active when this does not compromise truthfulness or
individual-rationality. Finally, suppose for contradiction that $p'=p_i(\hat{a}_i,\hat{d}_i,\theta_{-i},\omega)>
\max(\check{q}(\hat{a}_i,\hat{d}_i,\theta_{-i},\omega),\check{p}(\hat{a}_i,\hat{d}_i,\theta_{-i},\omega))$.
Then an agent with
$\max(\check{q}(\hat{a}_i,\hat{d}_i,\theta_{-i},\omega),\check{p}(\hat{a}_i,\hat{d}_i,\theta_{-i},\omega))<w_i<p'$
would prefer to bid
$\hat{w}_i=\check{q}(\hat{a}_i,\hat{d}_i,\theta_{-i},\omega),\check{p}(\hat{a}_i,\hat{d}_i,\theta_{-i},\omega))-\epsilon$
and avoid winning -- otherwise its payment would be greater than its value.
\end{proof}

\vspace{-0.1cm}
\subsection{Bounded Patience Is Required for Reasonable
Efficiency}
\label{sec:neg}

\chain\ depends on a maximal bound on patience 
used to calculate the admission price faced by a bidder on
entering the market with Eq.~(\ref{eqn:admit-price}).
To motivate this assumption about the existence of a 
maximal patience, we construct a 
simple environment in which the number of trades implemented by
a truthful, no-deficit DA can be made an arbitrarily small fraction of the
number of efficient trades with even a small number of bidders
having potentially unbounded patience. This illustrates that a bound
on bidder patience is required for dynamic DAs with reasonable
performance.

In achieving this negative result, we 
impose the additional requirement of {\em anonymity},
This anonymity property is already satisfied
by \chain, when coupled with matching rules that satisfy
anonymity, as is the case with all the rules presented in
Section~\ref{sec:practical}. 
In defining anonymity, extend the earlier definition
of a dynamic DA, $M=(\pi,x)$, so that allocation policy
$\pi=\{\pi^t\}^{t\in T}$ defines the {\em probability} $\pi^t_i(\theta^{\leq
t})\in[0,1]$ that agent $i$ trades in period $t$ given reports $\theta^{\leq
t}$. Payment, $x=\{x^t\}^{t\in T}$, continues to define the payment
$x_i^t(\theta^{\leq t})$ by agent $i$ in period $t$, and is a
random variable when the mechanism is stochastic.
\begin{definition}[anonymity]
A dynamic DA,  $M=(\pi,x)$ is anonymous if allocation policy
$\pi=\{\pi^t\}^{t\in T}$ defines probability of trade 
$\pi_i^t(\theta^{\leq t})$ in each period $t$ that is 
independent of identity $i$ and invariant to a permutation of
$(\theta^{\leq t}\setminus i)$ and if the payment $x^t_i(\theta^{\leq t})$,
contingent on trade by agent $i$, is independent of identity $i$ and invariant to
a permutation of $(\theta^{\leq t}\setminus i)$.
\end{definition}

We now consider the following simple environment. Informally, there
will be a random number of high-valued phases in which bids and asks
have high value and there might be a single bidder with patience that
exceeds that of the other bids and asks in the phase. These
high-valued phases are then followed by some number, perhaps zero, of
low-valued phases with bounded-patience bids and asks.
Formally, there are $T_h\geq 1$ {\bf
high-valued} phases (a random variable, unknown to the auction), each 
of duration $L\geq 1$ periods, indexed $k\in\{0,1,\ldots,T_h-1\}$ and
each with: 
\begin{itemize}
\item $N$ or $N-1$ bids with type $(1+kL,(k+1)L,v_H)$,
\item 0 or 1 bids with type $(1+kL,\overline{d},\alpha v_H)$ for some
mark-up parameter, $\alpha>1$ and some high-patience parameter, 
$\overline{d}\in T$,
\item $N$ asks with type $(1+kL,(k+1)L,-(v_H-\epsilon))$,
\end{itemize}
followed by some number (perhaps zero) of
{\bf low-valued} phases, also of duration $L$, and 
indexed $k\in\{T_h,\ldots,\infty\}$,  with:
\begin{itemize}
\item $N$ or $N-1$ bids with type $(1+kL,(k+1)L,v_L)$
\item $N$ asks with type $(1+kL,(k+1)L,-(v_L-\epsilon))$,
\end{itemize}

where $N\geq 1$, $0<v_L<v_H$, and bid-spread parameter $\epsilon>0$. 
Note that any phase can be the last phase, with no additional bids or
asks arriving in the future.
\begin{definition}[reasonable DA]
A dynamic DA is reasonable in this simple 
environment if there is {\bfseries some
parameterization} of new bids, $N\geq 1,$ and periods-per-phase, $L\geq
1$, for which it will execute {\bfseries at
least one trade} between new bids and new asks in {\bfseries each phase},  for any choice of
high value $v_H$, low value $v_L<v_H$, bid-spread $\epsilon>0$,
mark-up  $\alpha>1$, high patience $\overline{d}$.
\end{definition}

All of the 
dynamic DAs presented in Section~\ref{sec:practical} can be
parameterized to make them reasonable for a suitably large $N\geq 1$
and $L\geq 1$, and without the possibility of a bid with an unbounded
patience.   
\begin{theorem}
No strongly truthful, individual-rational, no-deficit, feasible, anonymous
dynamic DA can be reasonable when a bidder's patience can be unbounded.
\end{theorem}
\begin{proof}~Fix any $N\geq 1$, $L\geq 1$, and for the number of high-valued
phases, $T_h\geq 1$, set the departure of a high-patience
agent to $\overline{d}=(T_h+1)L$.
Keep $v_H>v_L>0$, $\epsilon>0$, and $\alpha>1$ as variables to be set within
the proof.
Assume a dynamic DA is reasonable, so that it selects at least one new
bid-ask pair to trade in each phase. 
Consider phase $k=0$ and with $N-1$ agents of types $(1,L,v_H)$, $N$ 
of type $(1,L,-(v_H-\epsilon))$ and 1 agent of patient type, 
$(1,(T_h+1)L,\alpha v_H)$.
If the patient bid deviates to $(1,L,v_H),$ then the bids are all
identical, and with probability at least $1/N$ the bid would win 
by anonymity and reasonableness. Also, by anonymity,
individual-rationality and no-deficit we have that the payment made by
any winning bid is the same, and must be
$p'\in[v_H-\epsilon,v_H]$. (If the payment had been less than this,
the DA would run at a deficit since the sellers require at least this
much payment for individual-rationality.) Condition now on the case
that the patient bid would win if it deviates and reports 
$(1,L,v_H)$. Suppose the bidder is truthful, reports 
$(1,(T_h+1)L,\alpha v_H)$ but does not trade in this phase. 
But, if phase $k=0$ is the last phase with new bids and asks, then the
bid will not be able to trade in the future and for
strong-truthfulness the DA would need to
make a payment of at least
$\alpha v_H - v_H=(\alpha-1)v_H$ in a later phase to prevent the
bid having a useful deviation to $(1,L,v_H)$ and 
winning in phase $k=0$. But, if:
\begin{align}
\label{eq:dp15}
N\epsilon&<(\alpha-1)v_H,
\end{align} 
then the DA cannot make this payment without failing no-deficit
(because $N\epsilon$ is an upper-bound on the surplus the auctioneer
could extract from bidders in this phase without violating
individual-rationality). 
We will later pick values of $\alpha, \epsilon$ and $v_H$, to satisfy
Eq.~(\ref{eq:dp15}). So, the bid must trade when it reports
$(1,(T_h+1)L,\alpha v_H)$, in the event that it would win with report 
$(1,L,v_H)$, as ``insurance" against this being the last phase with new
bids and asks. Moreover, it should trade for payment,
$p'\in[v_H-\epsilon,v_H]$, to ensure an agent with true type $(1,L,v_H)$
cannot benefit by reporting $(1,(T_h+1)L,\alpha v_H)$.
 
Now suppose that this was 
not the last phase with new bids and asks, and $T_h>1$. Now consider what would
happen if the patient bid in phase $k=0$ deviated and reported
$(1+T_hL,(T_h+1)L,v_L)$. As before, this bid would win with
probability at least $1/N$ by anonymity and reasonableness, but now
with some payment $p''\in[v_L-\epsilon,v_L]$.  Condition now 
on the case that the patient bid would win, both with a report 
of $(1,L,v_H)$ and with a report of
$(1+T_hL,(T_h+1)L,v_L)$. When truthful, it trades in phase $k=0$ with
payment at least $v_H-\epsilon$. If it had reported
$(1+T_hL,(T_h+1)L,v_L),$ it would trade in phase $k=T_h$ for payment at
most $v_L$. For strong truthfulness, the DA must 
make an additional payment to the patient agent of at least
$(v_H-v_L)-(v_H-(v_H-\epsilon))=v_H-v_L-\epsilon$. But, suppose
that the high and low values are such that,
\begin{align}
\label{eq:dp16}
(T_h+1)N\epsilon &< v_H-v_L-\epsilon.
\end{align}

Making this payment in this case would violate no-deficit,  because 
$(T_h+1)N\epsilon$ is an upper-bound on the surplus the auctioneer
can extract from bidders across all phases, including the current phase,
without violating individual-rationality. But now we can fix any $v_L>0$,
$\epsilon<v_L$ and choose $v_H>(T_h+1)N\epsilon+v_L+\epsilon$ to
satisfy Eq.~(\ref{eq:dp16}) and $\alpha>(N\epsilon/v_H)+1$ to satisfy
Eq.~(\ref{eq:dp15}). Thus, we have proved that no truthful dynamic DA
can choose a bid-ask pair to trade in period $k=0$. The proof can be
readily extended to show a similar problem with choosing a bid-ask
pair in any period $k<T_h$, by considering truthful type of
$(1+kL,(T_h+1)L,\alpha v_H)$.
\end{proof}

To drive home the negative result: notice that the number of efficient
trades can be increased without limit by choosing an arbitrarily large
$T_h$, and that no truthful, dynamic DA with these properties
will be able to execute even a
single trade in each of these $\{0,\ldots,T_h-1\}$ periods.
Moreover,  we see that only a {\em vanishingly small fraction of
high-patience agents} is required for this negative result. The proof
only requires that  at least one 
patient agent is possible in all of the high-valued phases.

\vspace{-0.2cm}
\section{Experimental Analysis}
\label{sec:empirical}

In this section, we evaluate in simulation each of the \chain-based DAs 
introduced in Section~\ref{sec:practical}.  We measure the allocative
efficiency (total value from the trades),  net efficiency (total value
discounted for the revenue that flows to the auctioneer), 
and revenue to the auctioneer. All values are normalized by the total
offline value of the optimal matching.

For comparison we also implement several 
other matching schemes: the truthful,
surplus-maximizing matching algorithm presented by Blum 
\etal~\citeyear{blum06}, an untruthful greedy matching
algorithm using truthful bids as input to provide an
upper-bound on performance, and an untruthful DA populated with
simple adaptive agents that are modeled after the Zero-intelligence Plus
trading algorithm that has been leveraged in the study of static
DAs~\shortcite{cliff:simple,preist:zip}. 

\subsection{Experimental Set-up}

Traders arrive to the market as a Poisson stream to exchange a single
commodity at discrete moments. This is a standard model of arrival
in dynamic systems, economic or otherwise.
Each trader, equally likely to be a buyer or seller, arrives after the
previous with an exponentially distributed delay, with probability density 
function (pdf):
\begin{align}
f(x)=\lambda e^{-\lambda x}, \ \ \ \ x\ge 0,
\end{align}
where $\lambda>0$ represents the arrival intensity in agents per
second. Later we present results as the {\em interarrival time}, 
$\frac{1}{\lambda},$ is varied between 0.05 and 1.5; i.e., as the
arrival intensity is varied between 20 and $\frac{2}{3}$.
A single trial continues until at least 5,000 buyers and 5,000 sellers have
entered the market. In our experiments we vary the maximal patience
$K$ between 2 and 10.
For the distribution on an agent's activity
period (or patience, $d_i-a_i$), we consider both a uniform
distribution with pdf:
\begin{align}
f(x)=\frac{1}{K}, \ \ \ \ x\in[0,K],
\end{align}
and a truncated exponential distribution with pdf:
\begin{align}f(x)=\alpha e^{-\alpha x}, \ \ \ \ x\in[0,K],\end{align}
where $\alpha=-\ln(0.05)/K$ so that 95\% of the underlying exponential
distribution is less than the maximal patience. 
Both arrival time and activity duration are rounded to the nearest integral
time period.
A trader who arrives and departs during the same period is assumed to 
need an immediate trade and is active for only one period.

Each trader's valuation represents a sample drawn at its arrival from a 
uniform distribution with spread 20\% about the current mean
valuation. (The value is positive for a bid and negative for an ask.) 
To simulate market volatility, we run experiments that 
vary the average valuation using Brownian motion, a common model for valuation 
volatility upon which many option pricing models are 
based~\shortcite{copeland:fintheory}.
At every time period, the mean valuation randomly increases or decreases by
a constant multiplier, $e^{\pm \gamma},$ where $\gamma$ is the approximate
volatility and varied between 0 and 0.15 in our experiments.

We plot the mean efficiency of 100 runs for each experiment, with the
same sets of bids and asks used across all double auctions. All
parameters of an auction rule are reoptimized for each market
environment; e.g., we can find the optimal fixed price and the optimal
smoothing parameters offline given the ability to sample from the
market model.

\subsection{Chain Implementation}

We implement \chain\ for 
the five price-based matching rules (history-clearing, history-median, history-McAfee, history-EWMA,
and fixed-price) and the three competition-based matching rules (McAfee, active-McAfee,
and windowed-McAfee).

The price-based implementations keep a fixed-size set of 
the most recently expired, traded, or priced-out offers, $H^t$.
Offers priced-out by their admission prices are inserted into $H^t$ 
prior to computing $p^t$.
The history-clearing metric computes a price to maximize the number of 
trades to agents represented by $H^t$ had they all been contemporary. 
The history-median metric chooses the price to be 
the median of the absolute valuation of the offers in $H^t$.
The history-McAfee method computes the ``McAfee price" for the
scenario where all agents represented by $H^t$ are simultaneously present. 
The EWMA metric computes an exponentially-weighted average of bids in the
order that they expire, trade, or price out.
The simulations initialize the price to the average of the mean buy and
sell valuations.
If two bids expire during the same period, they are included in arbitrary
order to the moving average.

None of the metrics require more than one parameter, which is
optimized offline with access to the model of the market environment. 
Parameter optimization proceeds by uniformly sampling the parameter range, smoothing the
result by averaging each result with its immediate neighbors.  The
optimization repeats twice more over a narrower range about the 
smoothed maximum, returning the parameter that maximizes (expected) allocative
efficiency.
None of the price-based methods
appeared to be sensitive to small ($<$10\%) changes in the size of $H^t$.
With most simulations, the window size was chosen to be about 150 offers.
For EWMA, the smoothing factor was usually chosen to be around 0.05 or lower.
The windowed-McAfee matching rule, however, was extremely sensitive to 
window size for simulations with volatile valuations, and the search process
frequently converged to suboptimal local maxima.

The admission price in the price-based methods is computed by first 
determining whether {\sc Match} would check the value of the bid
against bid price if the bid had arrived in some earlier period $t'$.
Rather than simulate the entire {\sc Match} procedure, it is sufficient
to determine the probability $\rho_i$ of this event. This is determined
by checking the construction of the strong no-trade sets in that
earlier period.
If $\mathrm{SNT}^{t'}$ contains non-departing buyers (sellers), then the 
probability that an additional  seller (buyer) would be examined is 1
and $\rho_i=1$.
Otherwise the probability is equal to the ratio
of the number of bids (asks) examined not included in $\mathrm{SNT}^t$
and one more than the total number of bids (asks) present.
Finally, with probability $\rho_i$ the
price the agent would have faced in period $t'$ is 
defined as $p^{t'}$ ($-p^{t'}$ for sellers), and otherwise it is
$-\infty$. Here, $p^{t'}$ is the history-dependent price defined
in period $t'$.

The competition-based matching rules price out all non-trading bids at the end of
each period in which trade occurs (because of the definition of strong no-trade
in that context). The admission prices are calculated by considering
the price that a bid (ask) would have faced in some period $t'$ before
its reported arrival. In such a period, the price for a bid (ask) is 
determined by inserting an additional bid (ask) with valuation $\infty$ (0)
and applying the competition-based matching rule to that (counterfactual) state.
From this we determine whether the agent would win for its reported value,
and if so what price it would face.

\subsection{Optimal Offline Matching}

We use a commercial integer program solver (CPLEX\footnote{www.ilog.com}) to
compute the optimal offline solution, i.e. with complete knowledge
about all offers received over time. In determining the offline solution we enforce the
constraint that a trade can only be executed if the activity periods of both 
buyer, $i,$ and seller, $j,$ overlap,
\begin{align}\label{eqn:overlap}
(a_i \le d_j) \ \wedge \ (a_j \le d_i)
\end{align}

An integer-program formulation to maximize total value is:
\begin{align}
\max \sum_{(i,j)\in\mathit{overlap}} & x_{ij}(w_i+w_j)\\
\quad \mbox{s.t. } 0  \leq
\sum_{i:(i,j)\in\mathit{overlap}}x_{ij}&\leq 1, \ \forall j\in\mathit{ask}
\notag \\
0 \leq \sum_{j:(i,j)\in\mathit{overlap}}x_{ij}&\leq 1, \ \forall i\in\mathit{bid}
\notag \\
x_{ij}&\in\{0,1\}, \ \forall i, j, \notag
\end{align}
where $(i,j)\in\mathit{overlap}$ is a bid-ask pair that could
potentially trade because they have overlapping arrival 
and departure intervals satisfying Eq.~(\ref{eqn:overlap}). 
The decision variable $x_{ij}\in\{0,1\}$
indicates that bid $i$ matches with ask $j$.
This provides the optimal, offline allocative efficiency.

\subsection{Greedy Online Matching}

We implement a greedy matching algorithm that immediately matches 
offers that yield non-negative budget surplus. This is a non-truthful
matching rule but  provides an 
additional comparison point for the efficiency of the other matching schemes.
During each time period, the greedy matching algorithm orders active bids
and asks by their valuations, exactly as the McAfee mechanism does, and matches
offers until pairs no longer generate positive surplus.
The algorithm's performance allows us to infer the number of
offers that the optimal matching defers before matching and the amount
of surplus lost by the McAfee method due to trade reduction and 
due to the additional constraint of admission pricing.

\subsection{Worst-Case Optimal Matching}

Blum \etal~\citeyear{blum06} 
derive a mechanism equivalent to our fixed-price matching 
mechanism, except that the price used is chosen from the cumulative distribution
\begin{align}\label{eqn:bsz-price}
D(x) = \frac{1}{r\alpha}\ln\left(\frac{x-w_{min}}{(r-1)w_{min}}\right),
\end{align}
where $r$ is the fixed point to the equation
\begin{align}
r = \ln\left(\frac{w_{max}-w_{min}}{(r-1)w_{min}}\right),
\end{align}
and $w_{min}\geq 0$ and $w_{max}\geq 0$ are the minimum and maximum absolute 
valuations of all traders in the market.
For our simulations, we give the mechanism the exact knowledge of the minimum
and maximum absolute valuations for each schedule.
Blum \etal~\citeyear{blum06} show that this method
guarantees an expected competitive ratio of
$\max(2,\ln(w_{max}/w_{min}))$ with respect to the optimal offline
solution in an adversarial setting. We were interested to see how will this performed in
practice in our simulations.

\subsection{Strategic Open-outcry Matching: ZIP Agents}

To compare \chain\ with the existing literature on
continuous double auctions, we 
implement a DA that in every period
sorts all active offers and matches the highest valued bids with the lowest 
valued asks so long as the match yields positive net surplus.  
The DA prices each trading pair at the mean of the pair's declared valuations.
{\em Since the trade price depends on a bidder's declaration,
the market does not support truthful bidding strategies.
We must therefore adopt a method to simulate the behavior of bidding
agents within this simple open outcry market.}

For this, we randomly assign each bid
to one of several ``protocol agents" that each use a modified 
ZIP trading algorithm, as initially presented by Cliff and 
Bruten~\citeyear{cliff:simple} and improved upon by Preist and van
Tol~\citeyear{preist:zip}. 
The ZIP algorithm is a common benchmark used to compare learned bidding 
behavior in a simple double-auction trading environment in which
agents are present at once and adjust their bids in seeking a
profitable trade.
We adapt the ZIP algorithm for use in our dynamic environment.

In our experiments we consider five of these protocol agents. 
New offers are assigned uniformly
at random to a protocol agent, which remains persistent throughout
the simulation. Each offer is associated with a patience category, $k\in
\{\mbox{low, medium, high}\}$, defined to evenly partition the
range of possible offer patience.
Each protocol agent, $j$, is defined with parameters
$(r_j,\beta_j,\gamma_j)$ and 
maintains a {\em profit margin}, $\mu_j^k$, on each patience category $k$.
Parameters $(\beta_j,\gamma_j)$ control the
adaptivity of the protocol agent in how it adjusts the target profit
margin on an individual offer, with $\beta_j\sim U(0.1,0.2)$ defining
the {\em offer-level learning rate} and $\gamma_j\sim U(0.2, 0.8)$ defining
the {\em offer-level damping factor}.
Parameter $r_j\in[0,1]$ is the {\em learning rate} adopted for updating
the profit margins. 

The protocol agents are trained over 10 trials and their final
performance is measured in the 11th trial. The learning rate 
decreases through the training session and depends on the
initial learning rate $r^0_j$ and the adjustment rate $r^+_j$.
In period $t\in\{1,\ldots,t^k_{\mbox{\tiny end}}\}$
of trial $k\in \{1,\ldots,T+1\}$, where $T=10$ is the number of trials
used for training and $t^k_{\mbox{\tiny end}}$ is the number of
periods in trial $k$, the learning rate is defined as:
\begin{align}
r_j := 1 - \left(r_j^0 + (k-1)r^+_j + \left(\frac{t}{t^k_{\mbox{\tiny
end}}}\right)^2r^+_j\right)
\end{align}
where $r^+_j=(1-r^0_j)/(T+1)$. We define $r^0_j:=0.7$. The effect of
this adjustment rule is that $r_j$ is initially 0.3, decreases
during training, and trends to
0.0 as $t\rightarrow t_{\mbox{\tiny end}}$ in trial $k=11$.

Within a given trial, upon assignment of a new offer $i$ in patience category $k$, 
the protocol agent managing the offer initializes $(\mu_i(t), \delta_i(t)):= (\mu_j^k, 0)$, where
$\mu_i(t)$ represents the {\em target profit margin} for the offer and
$\delta_i(t)$ represents a {\em profit-margin correction} term. 
The target profit margin and the profit margin correction term are
adjusted for offer $i$ in subsequent periods while the bid remains active.

The target profit margin is used to define a bid price for the offer
in each period while it remains active:
\begin{align}
\hat{w}_i(t)& :=
w_i(1+\mu_i(t)).
\label{eq:bidprice}
\end{align} 

At the end of a period in which an offer matches or simply expires,
the profit margin $\mu_j^k$ for its patience category is updated
as:
\begin{align}
\mu^k_j:= (1-r_j)\mu^k_j + r_j \mu_i(t),
\end{align}
where the amount of adaptivity depends on the learning rate $r_j$. 
Because the profit margin on an offer decays over its lifetime, this
update adjusts towards a small profit margin if the offer expires or
took many periods to trade, and a larger profit margin otherwise.
The long-term learning of a protocol agent occurs through the profit
margin assigned to each patience category. 

At the start of a period each protocol agent also computes {\em target
prices} for bids and asks in each patience category. 
These are used  to drive an adjustment  in the target profit margin
for each active bid and ask. 
Target prices $\tau^k_b(t)$ and $\tau^k_s(t)$ are computed as:
\begin{align}
\tau^k_b(t) & := \left\{
\begin{array}{ll}
 (1+\eta) \max\limits_{i\in B^k(t-1)}\{\hat{w}_i(t-1)\} +\xi
&
\mbox{, if $0 > \!\!\!\max\limits_{i\in S(t-1)}\{\hat{w}_i(t-1)\} + \!\!\!
	           \max\limits_{i\in B^k(t-1)}\{\hat{w}_i(t-1)\}$
}
\\
(1-\eta) \max\limits_{i\in B^k(t-1)}\{\hat{w}_i(t-1)\} -\xi
& \mbox{,	    otherwise}
\end{array}
\right.
\end{align}
and,
\begin{align}
\tau^k_s(t) & := 
\left\{
\begin{array}{ll}
(1+\eta) \max\limits_{i\in S^k(t-1)}\{\hat{w}_i(t-1)\} +\xi
& \mbox{, if $0 >\!\!\! \max\limits_{i\in B(t-1)}\{\hat{w}_i(t-1)\} + \!\!\!
	           \max\limits_{i\in S^k(t-1)}\{\hat{w}_i(t-1)\}$}
\\
          (1-\eta) \max\limits_{i\in S^k(t-1)}\{\hat{w}_i(t-1)\} -\xi
&
\mbox{, otherwise}
\end{array}
\right.
\end{align}
where $\xi, \eta \sim U(0,0.05)$. Here, $B(t-1)$ and $S(t-1)$ denote
the set of active bids and asks in the market in period $t-1$ (defined
before
market clearing), and $B^k(t-1)$ and $S^k(t-1)$ denote the restrictions to
patience category $k$. The target price on a bid in category $k$ is set to
something slightly {\em greater} than the most competitive bid in the
previous round when that bid could not trade, and slightly less
otherwise. Similarly for the target price on asks, where 
these prices are negative, so that increasing the target price makes an
ask more competitive.

Target prices are used to adjust the target profit margin at the start
of each period on all
active offers that arrived in some earlier period,
where the influence of target prices is through the profit-margin
correction term:
\begin{align}
\mu_i(t) & := 
\frac{(\hat{w}_i(t-1) + \delta_i(t))}{w_i}  - 1,
\label{eq:dp18}
\end{align}
and the profit-margin correction term, $\delta_i(t)$, is defined in
terms of the target price $\tau^k_i(t)$ (equal to $\tau^k_b(t)$ if $i$
is a bid and $\tau^k_s(t)$ otherwise) as,
\begin{align}
\delta_i(t) & := \gamma_j \delta_i(t-1) + (1-\gamma_j) \beta_j (\tau^k_i(t)
						     -
\hat{w}_i(t-1)),
\label{eq:dp19}
\end{align}
where $\gamma_j$ and $\beta_j$ are the offer-level learning rates and
damping factor. The value $w_i$ and the ``-1" term in 
Eq.~(\ref{eq:dp18}) provide
normalization. Eq.~(\ref{eq:dp19}) is the Widrow-Hoff \shortcite{hassoun:fnn}
rule, designed to minimize the least mean square error in the 
profit margin and adopted here to mimic earlier ZIP designs.

\subsection{Experimental Results}

Our experimental results show that market conditions drive DA choice.
We compare allocative efficiency, revenue, and net efficiency. 
All results are averaged over 100 trials.
In experiments we found only minimal qualitative differences between the
use of the two patience distributions. The uniform patience
distribution provides a slight increase in
efficiency over result using exponential patience, caused by a larger 
proportion 
of patient agents which relaxes somewhat the admission-price constraint
in Eq.~(\ref{eqn:admit-price}).
For this reason we choose to report only results for
the uniform patience distribution. 

While the performance of all methods are summarized in 
Table~\ref{tab:revenue}, we omit the performance of some markets from
the plots to keep the presentation of results
as clear as possible. We do not plot the price-based results for the 
median- or clearing-based prices because the performance was
typically around that of the performance of \chain\ instantiated
on the history-EWMA price. We do not plot the
windowed-McAfee results because of inconsistent performance,
and in most cases, upon manual inspection, it was optimal to choose the smallest
possible window size, i.e. including only active bids and making it
equivalent to active-McAfee.

Our plots also leave out the performance of the Blum \etal~\citeyear{blum06}
worst-case optimal matching scheme because it was
dominated by the fixed-price \chain\ instantiation and in
many cases failed to yield any substantial surplus. We note here
that the modeling assumption made by Blum \etal~\citeyear{blum06} is
quite different than that in our work: they worry about performance in an adversarial environment while we consider 
probabilistic environments. Our fixed-price \chain\ mechanism 
operates essentially identically to the surplus-maximizing 
scheme of Blum \etal~\citeyear{blum06}, except that \chain\ can also use additional statistical
information to set the ideal price, rather than drawing the price from a 
distribution that is used to guarantee worst-case performance
against an adversary. We defer the results for the Blum \etal~\citeyear{blum06} scheme to Table~\ref{tab:revenue}.

Figures~\ref{fig:ivS}--\ref{fig:pvs} plot results from two sets of
experiments, one for high-patience/low-volatility and one for 
low-patience/high-volatility, as we vary the inter-arrival time (and
thus the arrival intensity), volatility and maximal patience. All
plots  are for allocative efficiency except Figure~\ref{fig:ivs}, where
we consider net efficiency.
Active-McAfee is included on Figure~\ref{fig:ivS}, but not on any other
plots because it did not improve upon the McAfee performance in the 
other environments. To emphasize: the results for {\em greedy} provide
an upper-bound on the best possible performance because this is
a non-truthful algorithm, simulated here with truthful inputs. 

\begin{figure}
\centerline{
	\psfig{figure=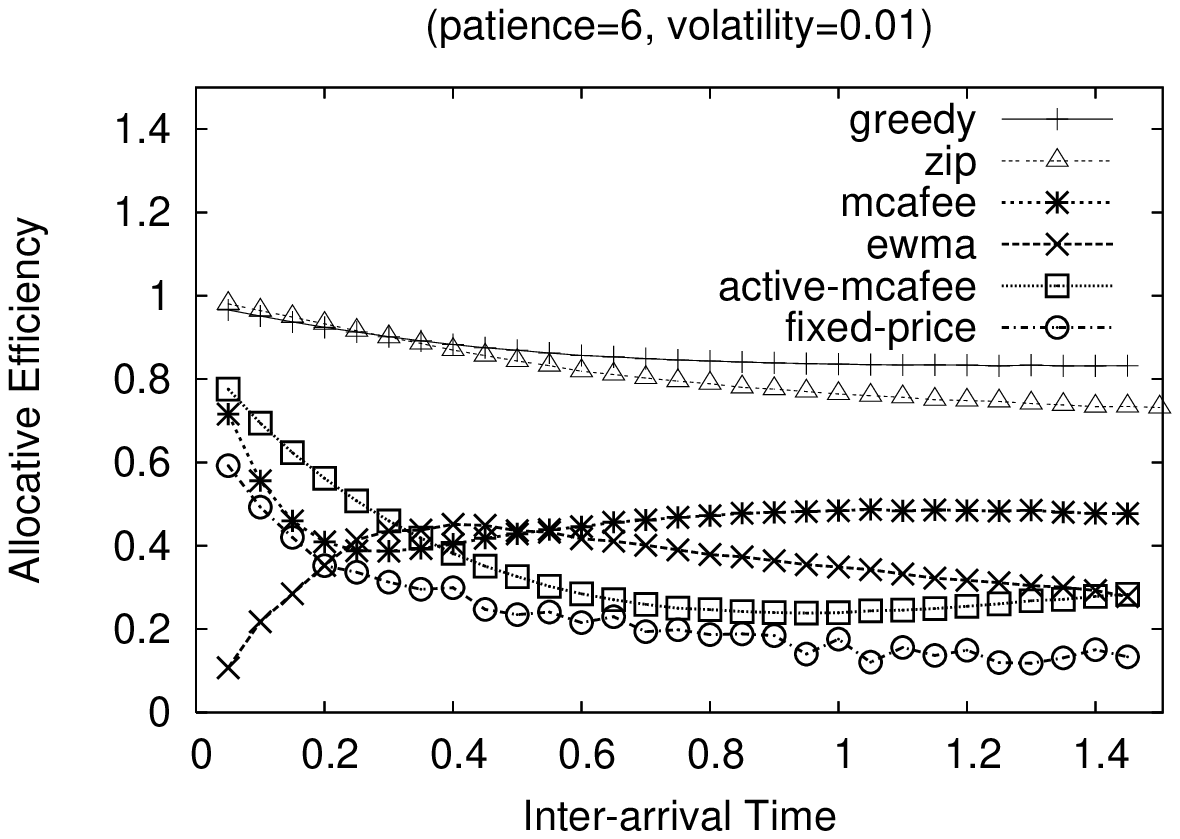,width=8.5cm}
	\psfig{figure=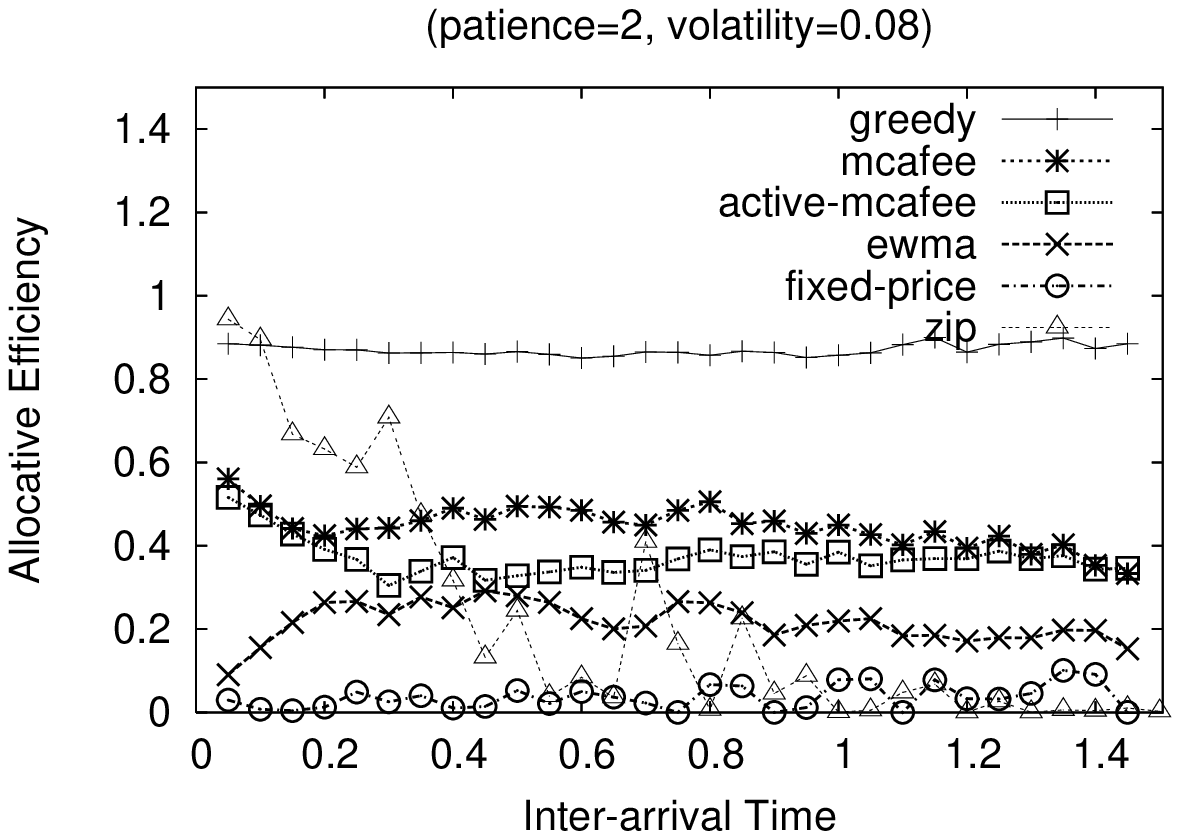,width=8.5cm}}
\caption{\small Allocative efficiency vs. inter-arrival time (1 / intensity)
for several DAs.  The left plot shows high-patience, low-volatility simulations,
whereas the right plots results from low-patience, high-volatility runs.
Both sets of experiments use uniform patience distributions.
\label{fig:ivS}}\vspace{-0.5cm}
\end{figure}

\begin{figure}
\centerline{
	\psfig{figure=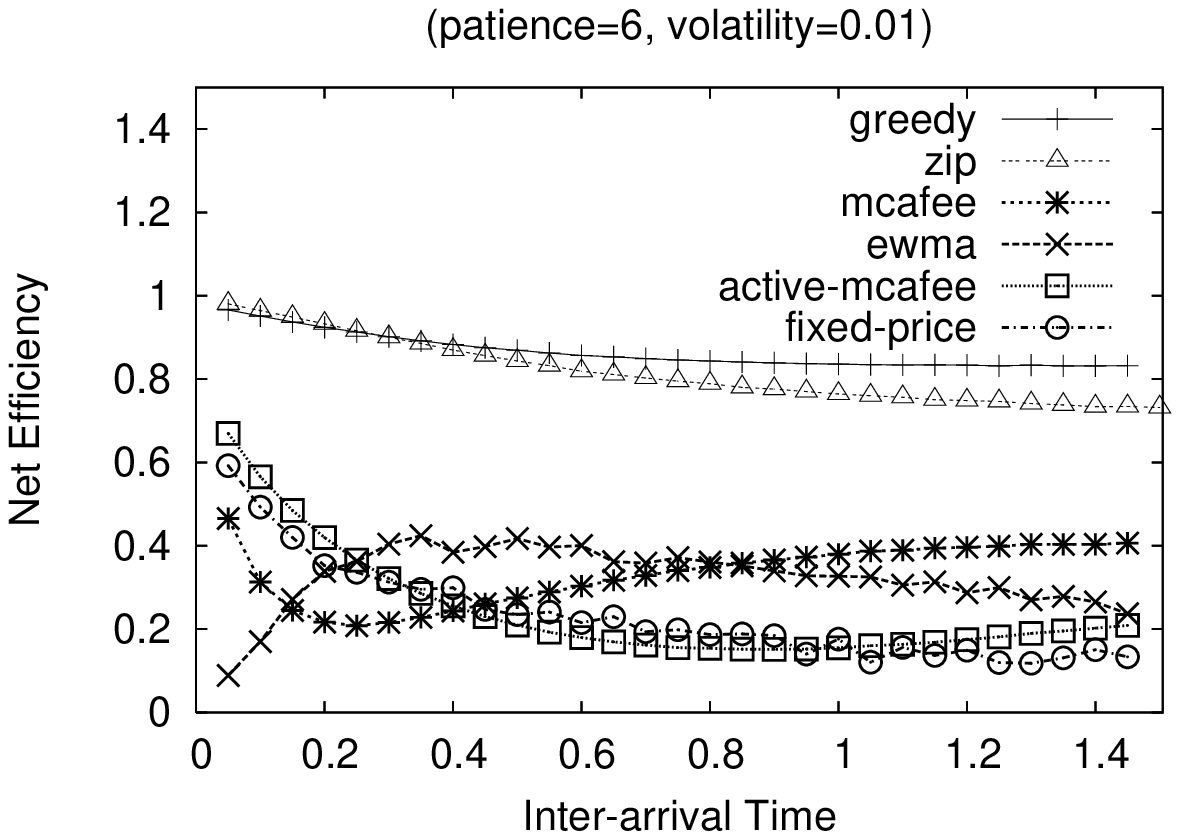,width=8.5cm}
	\psfig{figure=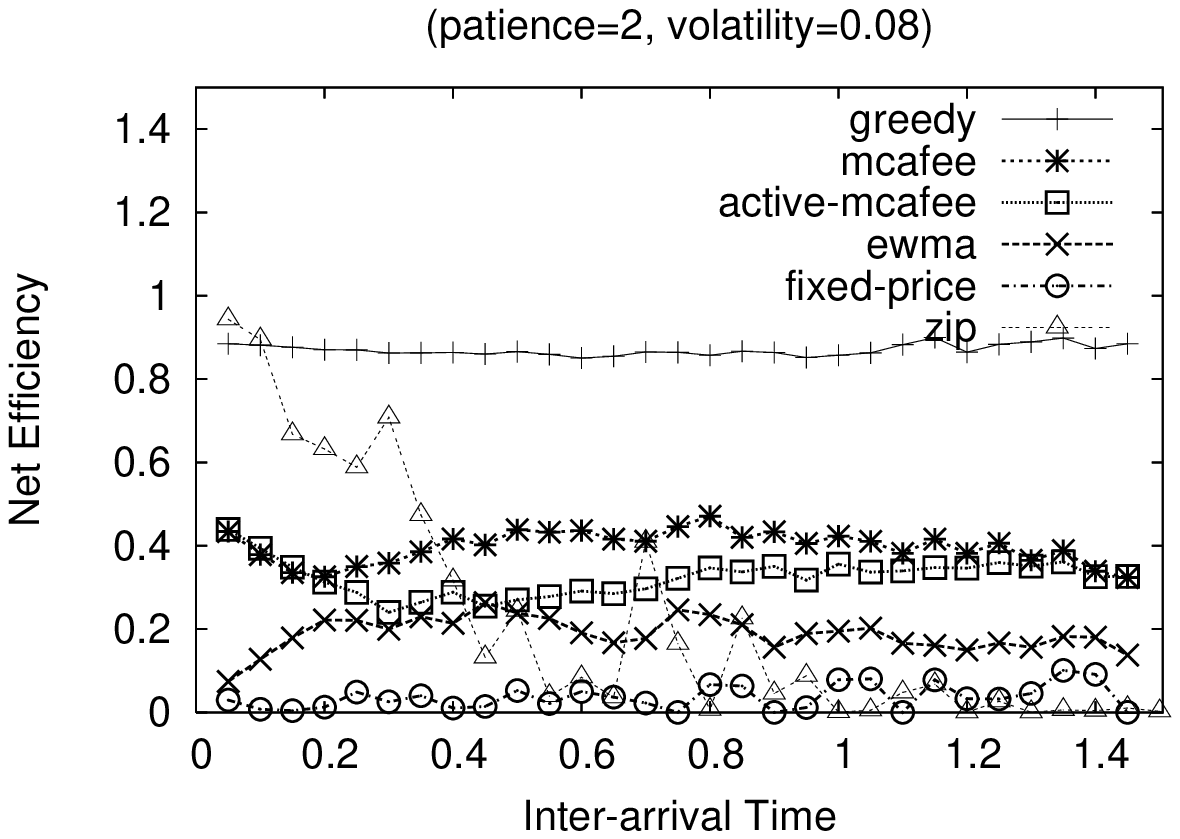,width=8.5cm}}
\caption{\small Net efficiency vs. inter-arrival time (1 / intensity) for several DAs.
The left plot shows high-patience, low-volatility simulations,
whereas the right plots results from low-patience, high-volatility runs.
Both sets of experiments use uniform patience distributions.
\label{fig:ivs}}\vspace{-0.5cm}
\end{figure}

In Figure~\ref{fig:ivS} (left) we see that 
from within the truthful DAs, the McAfee-based DA has the best 
efficiency for medium to low arrival intensities. 
There also is a general decrease in performance, relative to the
optimal offline solution, as the arrival intensity falls. This trend,
also observed with the greedy (non-truthful) DA, occurs because the
\chain\ scheme is myopic in that it matches as soon as the static
DA building block finds a match, while it is better to be less myopic
when arrival intensity is low.
The McAfee-based DAs are less sensitive to this than other methods because
they can aggressively update prices using the 
active traders. 
The price-based DAs experience inefficiencies due to the lag in price 
updates because they use only expired, traded, and priced-out offers
to calculate prices. 

For very high arrival intensity we see Active-McAfee dominates
McAfee. Active-McAfee smooths the price, which helps to
mitigate the impact of fluctuations in cost on the admission price
via Eq.~(\ref{eqn:admit-price}) in return for less responsiveness.
This is helpful in  ``well-behaved" markets with high arrival
intensity and low volatility but was not helpful in most environments we studied, where the additional responsiveness provided by the (vanilla) McAfee scheme
paid off.  

The ZIP market also has good performance in this
high-patience/low-volatility environment. The reason is simple: this
is an easy environment for simple learning agents, and the agents quickly
learn to be truthful. We emphasize that these ZIP market results should be treated with caution and are certainly optimistic. This is because 
the  ZIP agents are   not programmed to consider timing-based manipulations. 
The effect in this environment 
is that the ZIP market tends to operate as if a truthful market, but without
the cost of imposing truthfulness explicitly via market-clearing rules.
By comparison the \chain\
auctions are fully strategyproof, to both value and temporal
manipulations.

Compare now with Figure~\ref{fig:ivS} (right), which is for low
patience and high volatility. Now we see that McAfee dominates across
the range of arrival intensities. Moreover, the performance of ZIP 
is now quite poor because the agents do not have enough time to adjust
their bids (patience is low) and high volatility makes this a more
difficult environment. 
With volatile valuations, the possibility of valuation swings leaves open 
the possibility of larger profits, 
luring agents to set wider profit margins, but only after the market
changes.
The ZIP agents also
have fewer concurrent competitive offers to use in setting useful price
targets during learning. 
As we might expect,
high volatility also negatively impacts the efficiency of the
fixed-price scheme. 

In Figure~\ref{fig:ivs} we see that the net efficiency  trends
are qualitatively similar except that 
the competition-based DAs such as McAfee fare less well in comparison with the
price-based DAs. The auctioneer accrues more revenue for competition-based 
matching rules such as McAfee because they often generate buy and
sell prices with a spread. Together with the competition-based schemes being
intrinsically more dynamic, this drives an increased price spread in
\chain\ via the admission price constraints. 
In Figure~\ref{fig:ivs} (left) we see
that the fixed-price scheme
performs well for high arrival intensity while
EWMA dominates for intermediate arrival intensities. The McAfee 
scheme is still dominant  for lower
patience and higher volatility (Figure~\ref{fig:ivs}, right).

To reinforce these observations, in  Table~\ref{tab:revenue} we present
the the net efficiency, allocative efficiency and (normalized) 
revenue across all arrival intensities (i.e. 
inter-arrival time from 0.05 to 1.5) and 
for both low and high volatility trials. All five price-based
methods, all three competition-based methods, and all three comparison
methods are included. 
We highlight the best performing competition-based method, price-based
method, as well as the performance of the ZIP market (skipping over the non-truthful,
greedy algorithm). We omit information about the  mean standard error for each measurement because in no case did this error 
exceed a tenth of a percent of 
the mean optimal surplus.
From within the truthful DAs, we see that the McAfee-based scheme dominates overall for both
allocative and net efficiency and both low and high volatility,
although EWMA competes with McAfee for net efficiency in low
volatility markets. Notice also the good performance of the ZIP-based
market (with the aforementioned caveat about the restricted strategy space)
at low volatilities. 

\begin{table}[tb] 
\begin{center}
\begin{tabular}{l|lll|lll} 
&\multicolumn{6}{c}{\bf scenario} \\
&\multicolumn{3}{c|}{low-volt/high-pat}&\multicolumn{3}{c}{high-volt/low-pat}\\
& net & alloc & rev & net & alloc & rev \\\hline

mcafee			& {\bf 0.33}	& {\bf 0.47}	&{\bf  0.14}
& {\bf 0.40}	& {\bf 0.45}	& {\bf 0.05}\\
active-mcafee	& 0.24	& 0.35	& 0.11	& 0.32	& 0.37	& 0.05\\
windowed-mcafee	& 0.24	& 0.26	& 0.02	& 0.21	& 0.23	& 0.03\\\hline
history-clearing& 0.33	& 0.34	& 0.01	& 0.17	& 0.17	& 0.01\\
history-ewma	& {\bf 0.33}	& {\bf 0.35}	& {\bf 0.03}	& {\bf
0.19}	& {\bf 0.22}	& {\bf 0.03} \\
history-fixed	& 0.23	& 0.23	& 0.00	& 0.04	& 0.04	& 0.00\\
history-mcafee	& 0.33	& 0.34	& 0.01	& 0.15	& 0.16	& 0.01\\
history-median	& 0.33	& 0.34	& 0.01	& 0.17	& 0.18	& 0.01\\\hline
blum \etal& 0.10	& 0.10	& 0.00	& 0.02	& 0.02	& 0.00\\
greedy			& 0.86	& 0.86	& 0.00	& 0.87	& 0.87	& 0.00\\
zip				& {\bf 0.82}	& {\bf 0.82}	& {\bf
0.00}	& {\bf 0.23}	& {\bf 0.23}	& {\bf 0.00}\\
\end{tabular} 
\end{center} 
\caption{\small Net efficiency, allocative efficiency and auctioneer
revenue (all normalized by the optimal value from trade), 
averaged across all arrival 
intensities (0.05--1.5) and for low and high value volatility. 
The best performing
competition-based, price-based and `other' (ignoring greedy, which
is not truthful) results are highlighted.
\label{tab:revenue}}
\vspace{-0.5cm}
\end{table}

Figure~\ref{fig:vvs} plots allocative efficiency versus volatility
for high patience (left) and low patience (right) and for fairly
low arrival intensity.
Higher volatility hurts all methods~-- especially the ZIP agents, 
which struggle to learn appropriate profit and price targets, probably due to
few opportunities to update prices for every individual offer.
The McAfee scheme fairs very well, showing
good robustness for both large patience and small patience
environments. The fixed-price scheme has the best performance when
there is zero volatility but its efficiency falls off extremely
quickly as volatility increases. 

\begin{figure}
\centerline{
	\psfig{figure=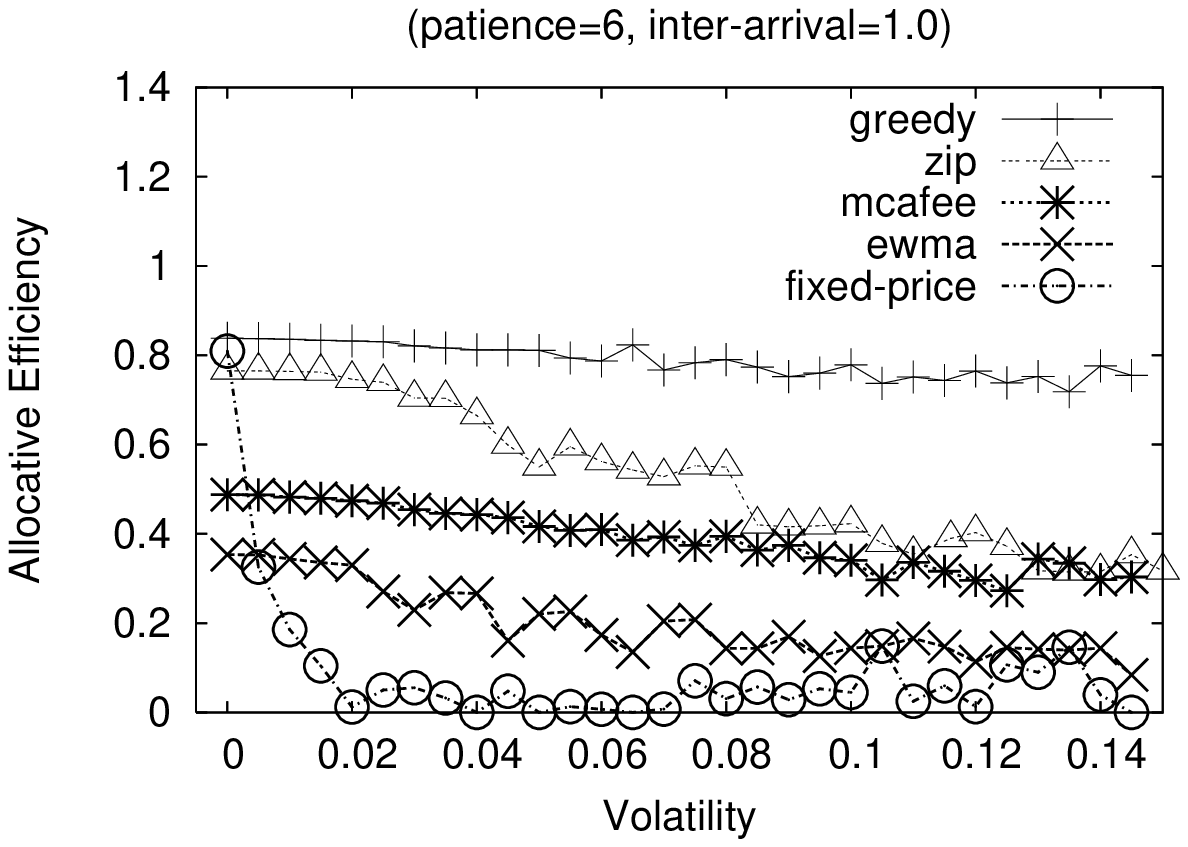,width=8.5cm}
	\psfig{figure=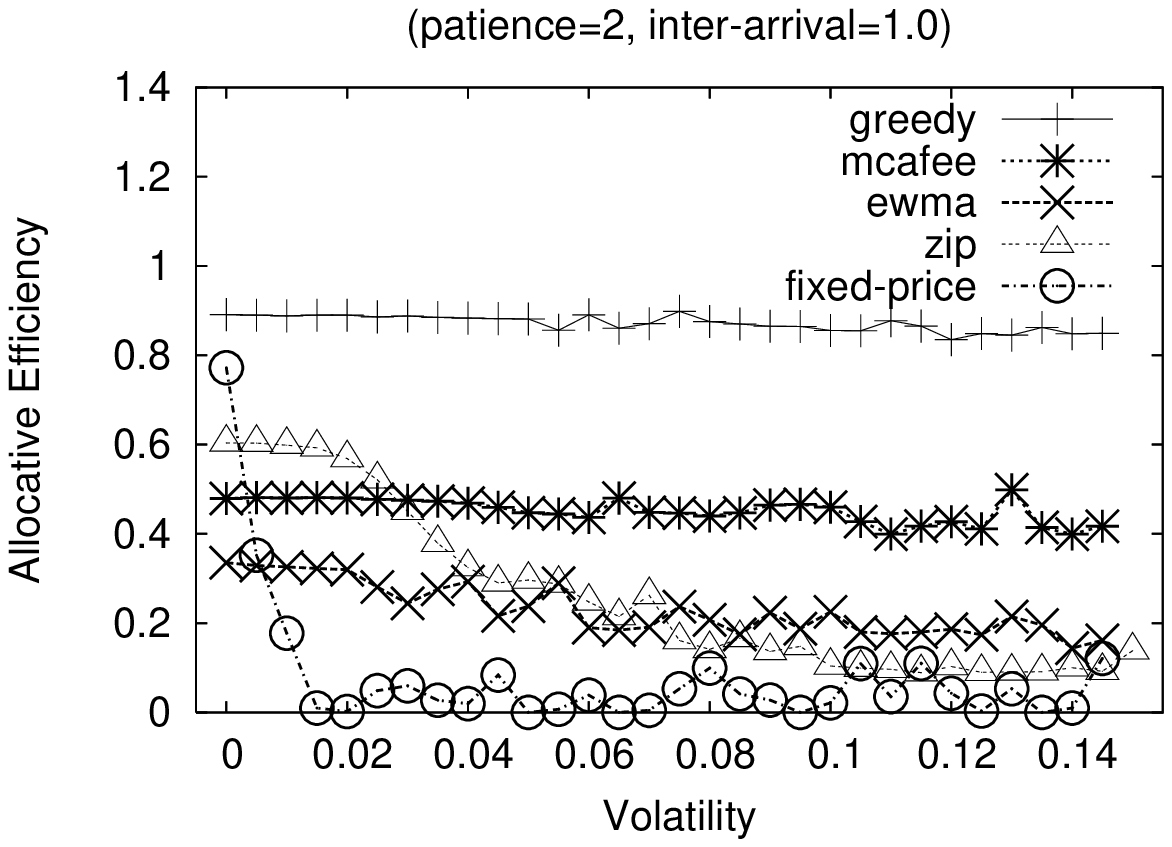,width=8.5cm}}
\caption{\small Allocative efficiency vs. volatility for several DAs
for a fairly low arrival intensity. The left plot is for large maximal 
patience and the right plot is for small maximal patience. 
Both sets of experiments use uniform patience distributions.
\label{fig:vvs}}\vspace{-0.6cm}
\end{figure}

We also consider the effect of varying maximal patience. This is
shown in Figure~\ref{fig:pvs}, with low volatility (left) and 
high volatility (right). Again, the McAfee
scheme is the best of the truthful DAs based on \chain. We also see that the performance
of ZIP improves as patience increases due to more opportunities for
learning.
Perversely, a larger patience can negatively affect the truthful
DAs. In part this is simply because the performance of greedy online
schemes, relative to the offline optimal, decreases as patience 
increases and 
the offline optimal matching is able to draw more benefit
from its lack of myopia. 
\begin{figure}
\centerline{
	\psfig{figure=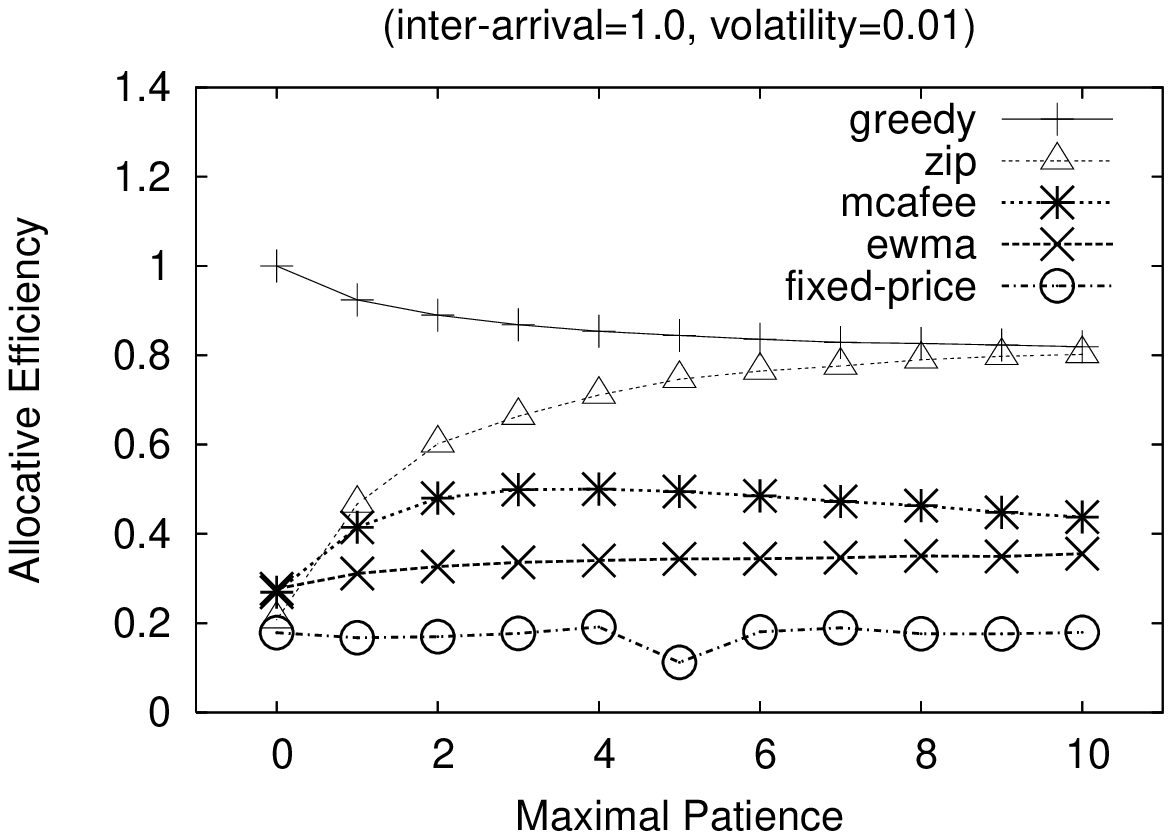,width=8.5cm}
	\psfig{figure=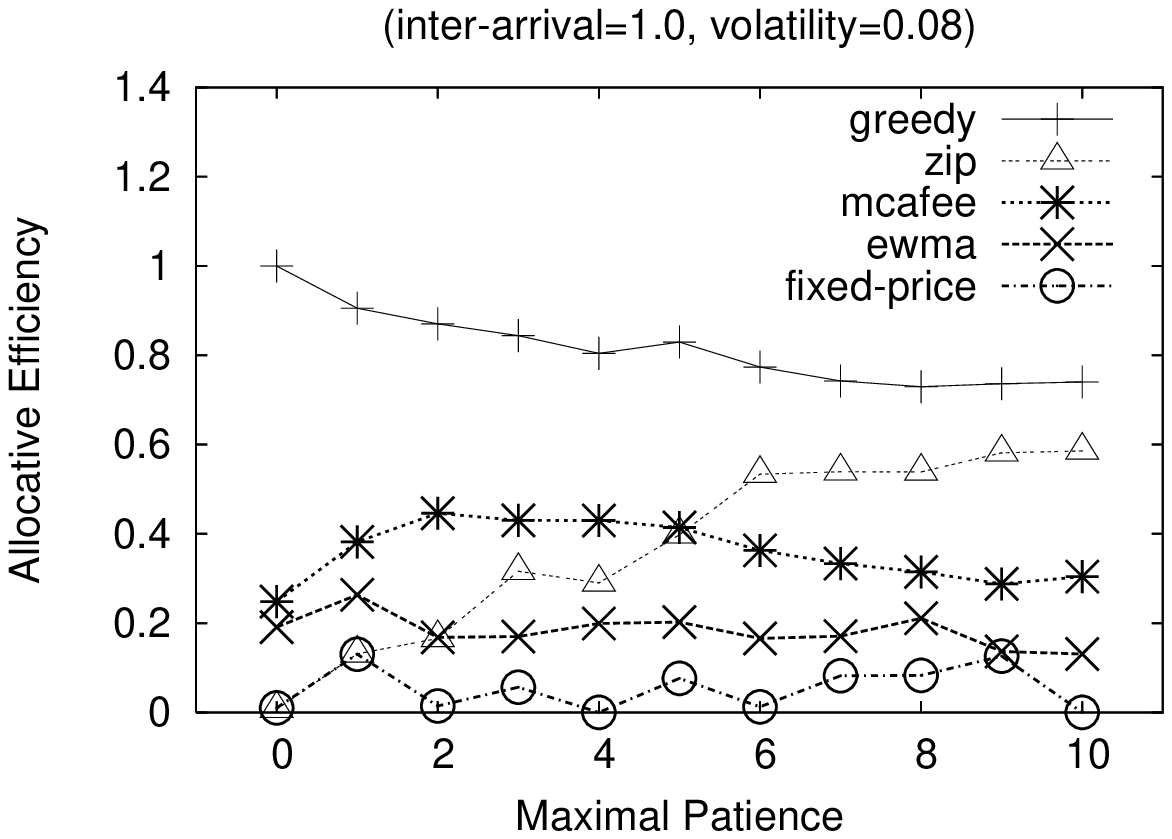,width=8.5cm}}
\caption{\small Allocative efficiency vs. maximal patience for
several DAs and fairly low arrival intensity.
 The left plot is for low volatility and the right plot is
for high volatility.
Both sets of experiments use uniform patience distributions.
\label{fig:pvs}}\vspace{-0.6cm}
\end{figure}

We also suspected another culprit, however.
The possibility
of the presence of patient agents requires the truthful DAs to include 
additional terms
in the $\max$ operator in Eq.~(\ref{eqn:admit-price}) to prevent
manipulations, leading to higher admission prices and less admitted
offers. 
To better understand this effect we experimented with delayed market
clearing in the McAfee scheme, where
the market matches agents only every $\tau$-th period (the ``clearing duration").
The idea is to make a tradeoff between using fewer admission prices and the
possibility that we will miss the opportunity to match some impatient
offers. 

Figure~\ref{fig:cvs} shows allocative efficiency when the matching mechanism 
clears less frequently and for different maximal patience,
$K$. Figure~\ref{fig:cvs} (left) is for low volatility. There we see
that the best clearing duration is roughly 1, 2, 3 and 4 for maximal
patience of $K\in\{4,6,8,10\}$ and that by optimizing the clearing
duration the performance of McAfee remains approximately constant as
maximal patience increases. In   Figure~\ref{fig:cvs} (right)
we consider the effect in a high volatility environment, with these
results averaged over 500 trials because the performance of the DA has
higher variance. We see a qualitatively similar trend, although higher
maximal patience now hurts overall and cannot be fully compensated for
by tuning the clearing duration.
\begin{figure}
\centerline{
	\psfig{figure=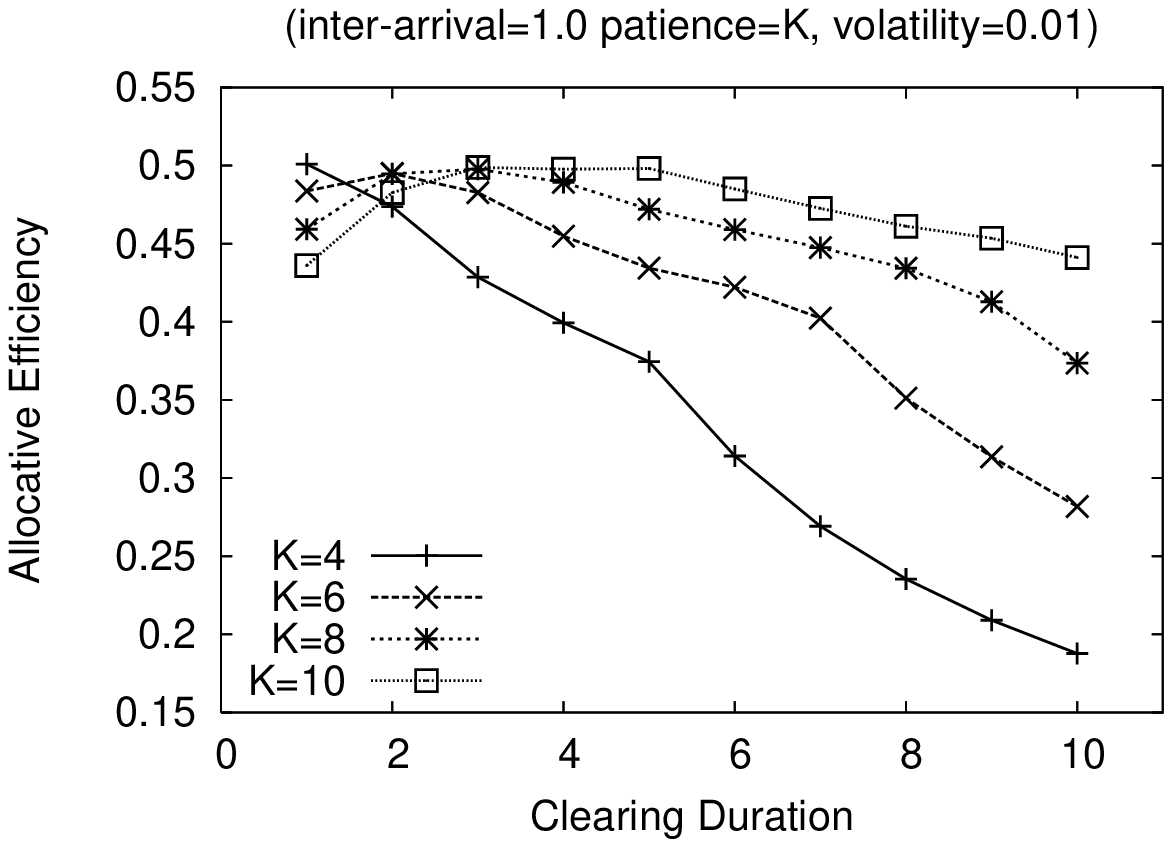,width=8.5cm}
	\psfig{figure=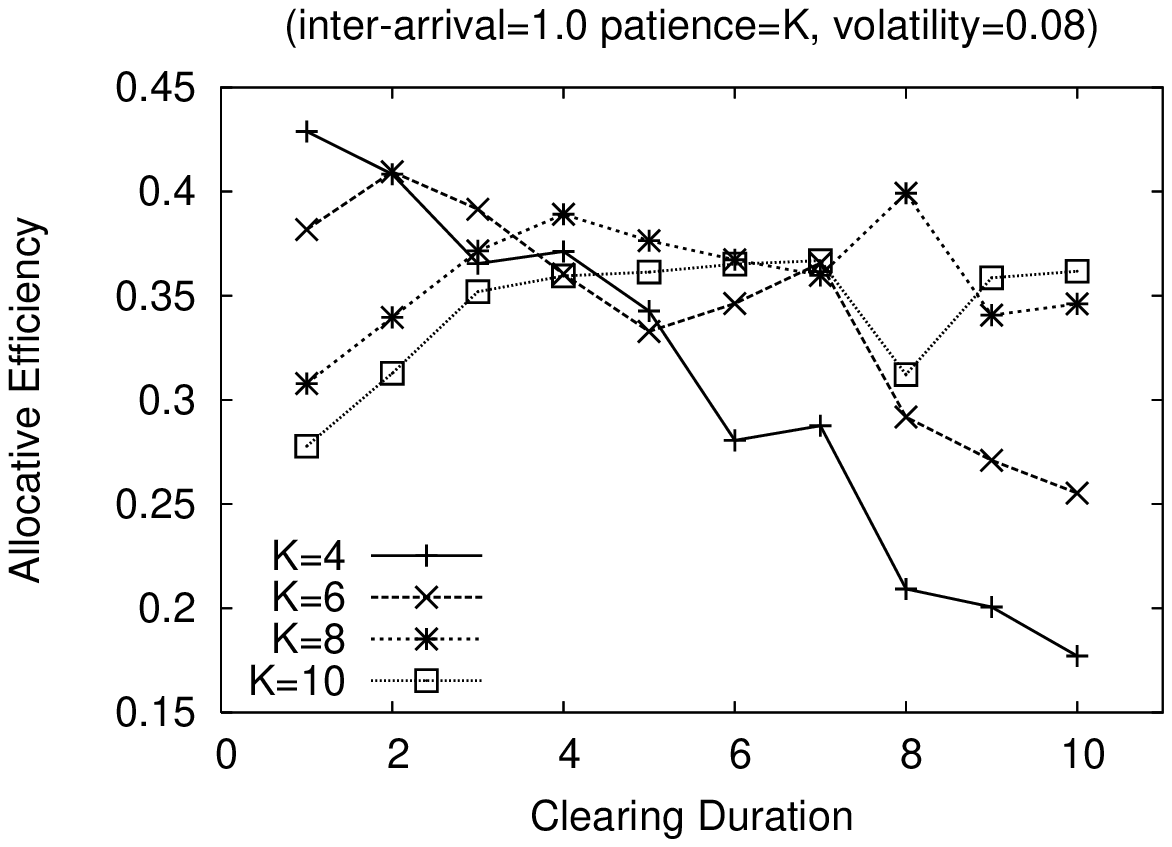,width=8.5cm}}
\caption{\small Allocative efficiency vs. clearing duration in the McAfee-based
\chain\ auction for fairly low arrival intensity and as maximal patience is varied from 4 to 10. The left plot is for low volatility and the right plot is for high volatility. Both sets of experiments use uniform patience distributions.
\label{fig:cvs}}\vspace{-0.5cm}
\end{figure}

\section{Related Work}
\label{sec:related}

Static two-sided market problems have been widely
studied~\shortcite{myerson:efficient,chatterjee87,satterthwaite89,yoon01,deshmukh02}.
In a classic result, 
Myerson and Satterthwaite proved that it is
impossible to achieve efficiency with voluntary participation and
without running a deficit, even relaxing dominant-strategy equilibrium to a
Bayesian-Nash equilibrium.
Some truthful DAs are known for static
problems~\shortcite{mcafee:double92,huang:double02,babaioff04a,babaioff05}.
For instance, McAfee introduced a
DA that sometimes forfeits trade in return for achieving
truthfulness. McAfee's auction achieves asymptotic
efficiency as the number of buyers and sellers increases.
Huang \etal\ extend McAfee's mechanism to handle agents exchanging 
multiple units of a single commodity.  Babaioff and colleagues
have considered extensions of this work to supply-chain and spatially
distributed markets.

Our problem is also similar to a traditional {\em continuous} double
auction (CDA), where buyers and sellers may at any time submit offers to a market that 
pairs an offer as soon as a matching offer is submitted. 
 Early work considered market efficiency
of CDAs with human experiments in labs~\shortcite{Smith62}, while 
recent work investigates the use of software agents to execute 
trades~\shortcite{Rust94,cliff:simple,gjerstad:priceform,tesauro:cda}.
While these markets have no dominant strategy equilibria, populations of
software trading agents can learn to extract virtually all available surplus,
and even simple automated trading strategies outperform human 
traders~\shortcite{das01agenthuman}. However, these studies of CDAs assume
that all traders share a known deadline by
which trades must be executed. This is quite different from our
setting, in which we have dynamic arrival and departure.

Truthful one-sided online auctions, in which agents arrive and depart across time, have
received some recent
attention~\shortcite{lavi00,hajiaghayi:online04,hajiaghayi05,porter04,lavi05}.
We adopt and extend the
monotonicity-based truthful characterization in the 
work of Hajiaghayi \etal~\citeyear{hajiaghayi05} in developing our
framework for truthful DAs. 
Our model of DAs must also address some of the
same constraints on timing that occur in Porter, 
Hajiaghayi, and Lavi and Nisan's work. 
In these previous works, the items were {\em
reusable} or {\em expiring} and could only be allocated in particular
periods. In our work we provide limited allowance to the match-maker,
allowing it to hold onto a seller's item until a matched buyer
is ready to depart (perhaps after the seller has departed).

The closest work in the literature is due to 
Blum \etal~\citeyear{blum06}, who present
online matching algorithms for the same dynamic DA model.
The main focus in their paper is on the design of matching algorithms 
with good worst-case performance in an adversarial setting, i.e. 
within the framework of competitive analysis. Issues related to
incentive compatibility receive less attention. One way in which
their work is more general is that they also study goals of {\em
profit} and {\em maximizing the number of trades}, in addition to 
the goal of maximizing social welfare that we consider in our 
work. However, the only algorithmic result that they present that is
truthful in our model (where agents can misreport arrival and
departure) is for the goal of social welfare. The DA that they
describe is an instance of \chain\ in which a {\em fixed 
price} is drawn from a distribution at the start of time, and used
as the matching price in every period. Perhaps unsurprisingly, given
their worst-case approach, 
we observe that their auction performs significantly
worse than \chain\ defined for a fixed price that is picked to
optimize welfare given distributional information about the domain.

\vspace{-0.3cm}
\section{Conclusions}
\label{sec:conclude}

We presented a general framework to construct algorithms to match buyers
with sellers in online markets where both valuation and activity-period 
information are private to agents.
These algorithms guarantee truthful dominant strategies by first
imposing a minimum admission price for each offer and then pricing and pairing the
offer at the first opportunity.
At the heart of the \chain\ framework lies a  pricing algorithm that must
for each offer either determine a price independent of any information 
describing the offer or choose to discard the offer.
The pricing algorithm should be chosen to match market conditions.
We present several examples of suitable pricing schemes, 
including fixed-price,  moving-average, and McAfee-based schemes. 

More often than not, we find that the competition-based scheme that employs
a McAfee-based rule to truthfully price the market delivers the best
allocative efficiency. For exceptionally low volatility and high
arrival intensity, we find that adaptive price-based schemes such as
an exponentially-weighted moving average (EWMA) and even fixed price schemes 
perform well. We see qualitatively similar results for net
efficiency, where the revenue that accrues to the auctioneer is
discounted, albeit that the price-based rules such as EWMA have improved
performance  because they have no price spread. 
The observations are rooted in simulations comparing the market efficiency
under each mechanism with the optimal offline solution.

Additionally, we compare the efficiency of our truthful markets with 
a fixed-price worst-case optimal scheme presented by 
Blum \etal~\citeyear{blum06}, 
a market of strategic agents using
a variant on the ZIP price update algorithm developed by Cliff and
Bruten~\citeyear{cliff:simple} for continuous double auctions,
and a non-truthful, greedy matching algorithm to provide an
upper-bound on performance. The best of our schemes yield
around 33\% net efficiency in low volatility, high patience
environments and 40\% net efficiency in high volatility, low patience
environments,  while the greedy bound suggests that as
much as 86\% efficiency is possible with non-strategic agents.
We note that the Blum \etal scheme, designed for adversarial settings, 
fairs poorly in our simulations ($< 10\%$). 

One can argue, we think convincingly, that truthfulness brings
benefits {\em in itself} in that it avoids the waste of costly
counterspeculation and promotes fairness in
markets~\shortcite{sandholm:compvic,abdulkadiroglu06}.
On the other hand, it is certainly of interest that the
gap between the efficiency of greedy matching with non-truthful 
matching and that of our truthful auctions is so large.
Here, we observe that the ZIP-populated
(non-truthful) markets achieve around 82\% efficiency in 
low volatility environments but collapse to around
23\% efficiency in high volatility environments. Based on this, one 
might conjecture
that designing for truthfulness is especially important in
badly behaved, highly volatile (``thin") 
environments but less important in well behaved, less volatile
(``thick") environments.

Formalizing this tradeoff between providing absolute truthfulness
and approximate truthfulness, and while considering the nature 
of the environment, is an interesting direction for
future work \shortcite<see paper by>{parkes01}. Given that
reporting of market statistics can be incorporated within our
framework (see Section~\ref{sec:report}), and given
that markets also play a role in information aggregation and
value discovery, future research should also consider this additional 
aspect of market design. Perhaps there is an  interesting tradeoff between efficiency, truthful
value revelation, and the process of information aggregation.

While the general \chain\ framework achieves good efficiency, further tuning seems 
possible. One direction is to adopt a {\em meta}-pricing scheme that
chooses, or blends, prices from competing algorithms.
Another direction is to consider richer temporal
models; e.g., the value of goods to agents might decay or grow over time to
better account for the time value of assets.
A richer temporal model might also consider the possibility of agents 
or the match-maker  taking short positions (including
short-term cash deficits) to increase trade.
It is also  interesting to extend our work to markets with non-identical
goods and more complex valuation models such as bundle trades~\shortcite{chu06,babaioff05,gonen07}, 
and to dynamic matching problems without
prices, such as an online variation of the classic ``marriage''
problem~\shortcite{gusfield:marriage}.

\vspace{-0.1cm}
\section*{Acknowledgments}

An earlier version of this paper appeared in the Proceedings 
of the 21st Conference on Uncertainty in Artificial Intelligence, 2005.
This paper further characterizes necessary conditions for truthful online trade;
truthfully matches offers using a generalized framework based upon an arbitrary
truthful static pricing rule; and compares the efficiency of our truthful 
framework to that achieved in non-truthful markets populated with strategic trading agents and with that of worst-case optimal double auctions.

Parkes is supported in part by NSF grant IIS-0238147 and an Alfred P. Sloan 
Fellowship and Bredin would like to thank the Harvard  School of
Engineering and Applied Sciences for hosting his sabbatical during which 
much this work was completed. Thanks also to the 
three anonymous reviewers, who provided
excellent suggestions in improving an earlier draft of this paper.

\section*{Appendix: Proofs}

\noindent {\bf Lemma 1} {\em Procedure {\sc Match} defines a valid strong no-trade construction.}
\vspace{0.3cm}

\begin{proof}~In all cases, $\mathrm{SNT}^t\subseteq \mathrm{NT}^t$. The set
$\mathrm{NT}^t$ is correctly constructed: equal to all remaining 
bids $b^t$ when $(j=0)$ in Case I, all remaining bids $s^t$ when
$(i=0)$ in Case II, and all remaining bids and asks otherwise. In each
case, no bid (or ask) in $\mathrm{NT}^t$ could have traded at any
price because there was no available bid or ask on the opposite of the
market given its order. 

In verifying strong no-trade (SNT) conditions (a) and (b), we proceed by
case analysis. 

\noindent
{\em Case I.} $(i\neq 0)$ and $(j=0)$. $\mathrm{NT}^t:=b^t$. 
\begin{enumerate}
\item[(I-1)] $\forall
k\in s^t \cdot (\hat{d}_k=t)$ and $\mathrm{SNT}^t:=b^t$. For SNT-a,
consider $l\in \mathrm{NT}^t$ with $\hat{d}_l>t$. If $l$ deviates and
$i$ changes but we remain in Case I then $\mathrm{NT}^t$ is unchanged
and still contains $l$. If $l$ deviates and $i\rightarrow 0$ then,
we go to Case III and $\mathrm{SNT}^t:=b^t\cup s^t$ and still contains
$l$. For SNT-b, consider $l\in\mathrm{SNT}^t$ that deviates with
$\hat{d}_l>t$. Again, either we remain in this case and
$\mathrm{SNT}^t$ is unchanged or $i\rightarrow 0$ and we go to Case
III. But now $\mathrm{SNT}^t$ still contains all $b^t$ and is
therefore unchanged for all agents with $\hat{d}_k>t$.
\item[(I-2)] Buyer $k\in b^t$ with $\hat{d}_k=t$ and $b_k\geq p^t$ and
$\mathrm{SNT}^t:=b^t$. For SNT-a, consider $l\in \mathrm{NT}^t$ with
$\hat{d}_l>t$. We remain in this case for any deviation by buyer $l$
because buyer $k$ will ensure $i\neq 0$, and so $\mathrm{SNT}^t$
remains unchanged and still contains $l$. For SNT-b, if
$l\in\mathrm{SNT}^t$ with $\hat{d}_l>t$ deviates we again remain in
this case and $\mathrm{SNT}^t$ is unchanged.
\item[(I-3)] Some seller with $\hat{d}_k>t$ and no buyer with
$\hat{d}_{k'}=t$ willing to accept the
price. $\mathrm{SNT}^t:=b^t\setminus\mathrm{checked}_B$. For SNT-a,
consider $l\in\mathrm{NT}^t$ with $\hat{d}_l>t$. First, suppose
$l\in\mathrm{checked}_B$ and $i\neq l$. If $l$ deviates but still has
$\hat{d}_l>t,$ then even if $i:=l$ then we remain in this case and $l$
does not enter $\mathrm{SNT}^t$. Second, suppose
$l\in\mathrm{checked}_B$ and $(i=l)$. If $l$ deviates but still has
$\hat{d}_l>t,$ then even if $(i=0)$ and $(j=0)$, we go to Case III 
and $\mathrm{SNT}^t=\emptyset$ and $l$ does not enter
$\mathrm{SNT}^t$. Third, suppose $l\notin\mathrm{checked}_B$ and
$\hat{d}_l>t$. Deviating while $\hat{d}_l>t$ has no effect and we
remain in this case and $l$ remains in $\mathrm{SNT}^t$. For SNT-b,
consider $l\in\mathrm{SNT}^t$ with $\hat{d}_l>t$, i.e. with $l\notin
\mathrm{checked}_B$. If $l$ deviates but $\hat{d}_l>t,$ then this has
no effect and we remain in this case and $\mathrm{SNT}^t$ remains
unchanged. 
\end{enumerate}

\noindent
{\em Case II}. $(j\neq 0)$ and $(i=0)$. $\mathrm{NT}^t:=s^t$. 
Symmetric with Case I.

\noindent
{\em Case III}. $(i=0)$ and $(j=0)$. $\mathrm{NT}^t:=b^t\cup s^t$. 
\begin{enumerate}
\item[(III-1)] $\forall k\in b^t \cdot (\hat{d}_k=t)$ but $\exists
k'\in s^t \cdot (\hat{d}_{k'}>t)$ and $\mathrm{SNT}^t:=b^t\cup
s^t$. For SNT-a, consider $l\in\mathrm{NT}^t$ with $\hat{d}_l>t$. This
must be an ask. If $l$ deviates but we remain in this case, then $l$
remains in $\mathrm{SNT}^t$. If $j:=l,$ then we go to Case II and
$\mathrm{SNT}^t:=s^t$ and $l$ remains in $\mathrm{SNT}^t$. For SNT-b,
consider $l\in\mathrm{SNT}^t$ with $\hat{d}_l>t$, which must be an
ask. If $l$ deviates but we remain in this case, $\mathrm{SNT}^t$
is unchanged. If $l$ deviates and $j:=l,$ then we go to Case II,
$\mathrm{SNT}^t:=s^t,$ and buyers $b^t$ are removed from
$\mathrm{SNT}^t$. But this is OK because all buyers depart in period
$t$ anyway.
\item[(III-2)] $\forall k\in s^t cdot (\hat{d}_k=t)$ but $\exists
k'\in b^t \cdot (\hat{d}_{k'}>t)$ and $\mathrm{SNT}^t:=b^t\cup
s^t$. Symmetric to Case III-1.
\item[(III-3)] $\forall k\in b^t \cdot (\hat{d}_k=t)$ and $\forall
k\in s^t cdot (\hat{d}_k=t)$. $\mathrm{SNT}^t:=b^t\cup s^t$. SNT-a and
SNT-b are trivially met because no bids or asks have departure past
the current period.
\item[(III-4)] $\exists
k\in b^t \cdot (\hat{d}_{k}>t)$ and $\exists
k'\in s^t \cdot (\hat{d}_{k'}>t)$ and $\mathrm{SNT}^t:=\emptyset$. For
SNT-a, consider $l\in\mathrm{NT}^t$ with $\hat{d}_l>t$. Assume that
$l$ is a bid. If $l$ deviates and $\hat{d}_l>t$ and $i=0$ then we
remain in this case and $l$ is not in $\mathrm{SNT}^t$. If $l$ deviates
and $\hat{d}_l>t$ but $i:=l,$ then we go to Case I and we are
necessarily in Sub-case (I-a) because $\hat{d}_l>t$ and there can be
no other bid willing to accept the price (else $i\neq 0$ in the first
place). Thus, we would have
$\mathrm{SNT}^t:=b^t\setminus\mathrm{checked}_B$ and $l$ would not be
in $\mathrm{SNT}^t$. For SNT-b, this is trivially satisfied because
there are no agents $l\in \mathrm{SNT}^t$.
\end{enumerate}
\end{proof}

\noindent {\bf Lemma 5} {\em The set of active agents (other than $i$) in period $t$ in \chain\
is independent of $i$'s report while agent $i$ remains active, and
would be unchanged if $i$'s arrival is later than period $t$.}
\vspace{0.3cm}

\begin{proof}~Fix some arrival period $\hat{a}_i$. Show for any $\hat{a}_i\geq a_i$,
the set of active agents in period $t\geq \hat{a}_i$ while $i$ is
active is the same as $A^t$ without agent $i$'s arrival until some
$a'_i>t$. Proceed by induction on the number of periods that $t$ is after
$\hat{a}_i$. For period $t=\hat{a}_i$ this is trivial. Now consider
some period $\hat{a}_i+r$, for some $r\geq 1$ and assume the inductive
hypothesis for $\hat{a}_i+r-1$. Since $i$ is still active then,
$i\in\mathrm{SNT}^{t'}$ for $t'=\hat{a}_i+r-1$, and therefore the
other agents in $\mathrm{SNT}^{t'}$ that survive into this period are
independent of agent $i$'s report by strong no-trade condition
(b). This completes the proof.
\end{proof}

\noindent {\bf Lemma 6} {\em The price constructed from admission price $\check{q}$ and
post-arrival price $\check{p}$ is value-independent and monotonic-increasing when
the matching rule in \chain\ is well-defined, the strong no-trade
construction is valid, and agent patience is bounded by $K$.}
\vspace{0.3cm}

\begin{proof}~First fix $\hat{a}_i,\hat{d}_i$ and $\theta_{-i}$. To show
value-independence (B1), first note that $\check{q}$ is
value-independent, since whether or not $i\in\mathrm{SNT}^t$ in some
pre-arrival period $t$ is value-independent by strong no-trade
condition (a) and price $z_i(H^t,A^t\setminus i,\omega)$ in such a
period is agent-independent by definition. Term $\check{p}$ is also
value-independent: the decision period
$t*$ to agent $i$, if any, is independent of $\hat{w}_i$ since the
other agents that remain active are independent of agent $i$ while it is
active by Lemma~\ref{lem:new}, and whether or not $i\in
\mathrm{SNT}^t$ is value-independent by strong no-trade (a);
and the price in $t*$ is value-independent when the set of other active
agents are value-independent.

Now fix $\theta_{-i}$ and show the price is monotonically-increasing
in a tighter arrival-departure interval (B2). First note that
$\check{q}$ is monotonic-increasing in
$[\hat{a}_i,\hat{d}_i]\subset[a_i,d_i]$ because an earlier $\hat{d}_i$
and later $\hat{a}_i$ increases the domain
$t\in[\hat{d}_i-K,\hat{a}_i-1]$ on which $\check{q}$ is defined. Fix
some $\hat{a}_i\geq a_i$. Argue the price
increases with earlier $d'_i\leq \hat{d}_i$, for any
$\hat{d}_i>\hat{a}_i$. To see this, note that either $\hat{d}_i<t*$
and so $p_i(\hat{a}_i,d'_i,\theta_{-i},\omega)=\infty$ for all
$d'_i\leq \hat{d}_i$, or $\hat{d}_i\geq t*$ and the price is constant
until $\hat{d}_i<t*$ at which point it becomes $\infty$. Fix some
$\hat{d}_i\geq a_i$. Argue the price increases with later $a'_i\geq
\hat{a}_i$, where $\hat{a}_i\geq \hat{d}_i-K$. First, while $a'_i\leq
t*,$ then $\check{p}$ is unchanged by Lemma~\ref{lem:new}. The
interesting case is when $a'_i>t*$, especially when
$\check{q}(\hat{a}_i,\hat{d}_i,\theta_{-i},\omega)<\check{p}(\hat{a}_i,\hat{d}_i,\theta_{-i},\omega)$.
By reporting a later arrival, the agent can delay its decision period
and perhaps hope to achieve a lower price. But, note that in this case
$t*\in[\hat{d}_i-K,a'_i-1]$ since $\hat{d}_i-K\leq \hat{a}_i$ and
$t*\in[\hat{a}_i,a'_i-1]$ and so
$\check{q}(a'_i,\hat{d}_i,\theta_{-i},\omega)\geq
\check{p}(H^{t*},A^{t*}\setminus i,\omega)$ because $\check{q}$ now
includes the price in period $t*$ since $i\notin \mathrm{SNT}^{t*}$ in
that pre-arrival period by Lemma~\ref{lem:new}. Overall, we see that
although $\check{p}$ may decrease, $\max(\check{q},\check{p})$ cannot
decrease.
\end{proof}

\noindent {\bf Lemma 7} {\em A strongly truthful, 
canonical dynamic DA must define price
$p_i(\hat{a}_i,\hat{d}_i,\theta_{-i},\omega)\geq
z_i(H^{t*},A^{t*}\setminus i,\omega)$ where $t*$ is the decision
period for bid $i$ (if it exists). Moreover, the bid must be
priced-out in period $t*$ if it is not matched.}
\vspace{0.3cm}

\begin{proof}~(a) First, suppose $z_i(H^{t*},A^{t*}\setminus i,\omega)>\hat{w}_i$ but
bid $i$ is not priced-out and instead survives as an active bid into
the next period. But with $i\notin\mathrm{SNT}^{t*}$, the set
of active bids in period $t*+1$ need not be independent of agent $i$'s
bid and the price $z_i(H^{t*+1},A^{t*+1}\setminus i,\omega)$ need
not be agent-independent. Yet, canonical rule (iii) requires that this
price be used to determine whether or not the agent matches, and so
the dynamic DA need not be truthful. (b) Now assume for contradiction that
$p_i(\hat{a}_i,\hat{d}_i,\theta_{-i},\omega)<z_i(H^{t*},A^{t*}\setminus
i,\omega)$. First, if $z_i(H^{t*},A^{t*}\setminus i,\omega)<\infty,$ then
an agent with value
$p_i(\hat{a}_i,\hat{d}_i,\theta_{-i},\omega)<w_i<z_i(H^{t*},A^{t*}\setminus
i,\omega)$ will report $\hat{w}_i=z_i(H^{t*},A^{t*}\setminus
i,\omega)+\epsilon$ and trade now for a final payment less than
its true value (whereas it would be priced-out if it reported its true
value). If $z_i(H^{t*},A^{t*}\setminus i,\omega)=\infty$, then
$p_i(\hat{a}_i,\hat{d}_i,\theta_{-i},\omega)<z_i(H^{t*},A^{t*}\setminus
i,\omega)$ implies that some bids will survive this period even though
they are priced-out by the matching rule and not in the
strong no-trade set. This compromises the truthfulness of the dynamic
DA, as discussed in part (a).
\end{proof}

\noindent {\bf Lemma 8} {\em A strongly truthful, canonical and individual-rational  dynamic 
DA must define price $p_i(\hat{a}_i,\hat{d}_i,\theta_{-i},\omega)\geq
\check{q}(\hat{a}_i,\hat{d}_i,\theta_{-i},\omega)$, and a bid with
$\hat{w}_i<\check{q}(\hat{a}_i,\hat{d}_i,\theta_{-i},\omega)$ must be
priced-out upon admission.}
\vspace{0.3cm}

\begin{proof}~Suppose $\hat{d}_i<\hat{a}_i+K$ so that $[\hat{d}_i-K,\hat{a}_i-1]$ is
non-empty. For $\hat{d}_i=\hat{a}_i+K-1$, when $t=\hat{d}_i-K$ is a
decision period (and $i\notin\mathrm{SNT}^t$), we have
\begin{align}
\label{eq:dp12}
p_i(\hat{a}_i,\hat{d}_i,\theta_{-i},\omega)\geq
p_i(\hat{d}_i-K,\hat{d}_i,\theta_{-i},\omega)&\geq 
z_i(H^t,A^t\setminus i,\omega),
\end{align}
where the first inequality is by monotonicity (B2) and the second
follows from Lemma~\ref{lem:dp5} since $\hat{d}_i-K$ is a decision
period, and would remain one with report
$\theta'_i=(\hat{d}_i-K,\hat{d}_i,w'_i)$ by Lemma~\ref{lem:new}. This
establishes $p_i(\hat{a}_i,\hat{d}_i,\theta_{-i},\omega)\geq
\check{q}(\hat{a}_i,\hat{d}_i,\theta_{-i},\omega)$ for
$\hat{d}_i=\hat{a}_i+K-1$. When $\hat{d}_i=\hat{a}_i+K-2$, then we
need Eq.~(\ref{eq:dp12}), and also when $t=\hat{d}_i-K+1$ is a
decision period (and $i\notin \mathrm{SNT}^t$) we have,
\begin{align}
\label{eq:dp13}
p_i(\hat{a}_i,\hat{d}_i,\theta_{-i},\omega)\geq
p_i(\hat{d}_i-K+1,\hat{d}_i,\theta_{-i},\omega)&\geq 
z_i(H^t,A^t\setminus i,\omega),
\end{align}
by the same reasoning as above. This generalizes to $d_i=a_i+K-r$ for
$r\in\{2,\ldots,K\}$ to establish  $p_i(\hat{a}_i,\hat{d}_i,\theta_{-i},\omega)\geq
\check{q}(\hat{a}_i,\hat{d}_i,\theta_{-i},\omega)$ for the general
case. To see the bid must be priced-out when
$\hat{w}_i<\check{q}(\hat{a}_i,\hat{d}_i,\theta_{-i},\omega),$ note
that if it were to remain active it could match in the matching rule
and by canonical (iii) need to trade, and thus fail
individual-rationality since the payment collected would be more than
the value.
\end{proof}

\bibliographystyle{theapa}
\bibliography{uai,master}

\end{document}